\documentclass[11pt]{article}
\usepackage{hyperref}
\usepackage{mathpazo}
\usepackage{amssymb,amsmath,amsthm}
\usepackage{epsfig}
\usepackage{tcolorbox}
\usepackage{enumitem}
\usepackage{booktabs}

\tcbset{ 
  halign=justify, 
  center, 
  colback=gray!10, 
  colframe=black, 
  boxrule=0.5pt 
}

\newcommand{\SD}{\mathop{\mathrm{SD}}}

\newtheorem{theorem}{Theorem}[section]
\newtheorem{proposition}[theorem]{Proposition}
\newtheorem{definition}[theorem]{Definition}

\newtheorem{lemma}[theorem]{Lemma}

\newtheorem{remark}[theorem]{Remark}

\def\lsoft{{l\kern-0.035cm\char39\kern-0.03truecm}}

\newcommand{\qedsymb}{\hfill{\rule{2mm}{2mm}}}
\renewenvironment{proof}[1][]{\begin{trivlist}
\item[\hspace{\labelsep}{\bf\noindent Proof#1:\/}] }{\qedsymb\end{trivlist}}

\def\calA{{\cal A}}
\def\calB{{\cal B}}
\def\calC{{\cal C}}

\def\calF{{\cal F}}

\def\calP{{\cal P}}

\def\R{\mathbb{R}}

\newcommand{\eps}{\epsilon}
\renewcommand{\epsilon}{\varepsilon}

\begin{document}

\title{{\bf Fair Allocation under Conflict Constraints \\ via Strong Colorability}}

\author{
Ishay Haviv\thanks{The Academic College of Tel Aviv-Yaffo, Tel Aviv, Israel.}
}

\date{}

\maketitle

\begin{abstract}
In the fair allocation problem under conflict constraints, the goal is to partition the vertices of a graph among agents in a fair manner, such that no two adjacent vertices are assigned to the same agent. We study this problem for agents with common preferences through the lens of three fairness criteria: stochastic-dominance envy-freeness up to one item for preference orders (SD-EF1), envy-freeness up to one item for monotone additive valuations (EF1), and envy-freeness up to one item from each side for general additive valuations (EF[1,1]). To do so, we introduce a hierarchy of variants of the strong chromatic number, a graph quantity introduced independently by Alon and Fellows in the early nineties. Our results reveal a close connection between fair allocation under conflict constraints and the first two levels of this hierarchy, providing a unified route to both existential and algorithmic results.

For SD-EF1, we fully characterize the number of agents needed to guarantee a fair allocation of a given graph for every common preference order. For EF1 and EF[1,1], we provide analogous sufficient conditions, extending a result on path graphs due to Equbal, Gurjar, Igarashi, Kumar, Manurangsi, Nath, Saxena, Vaish, and Yoneda. We also show that, unlike in the SD-EF1 setting, the sufficient conditions for EF1 and EF[1,1] are not necessary in general. Our framework yields existential and algorithmic consequences in terms of the maximum degree. We obtain that every graph with maximum degree $\Delta$ admits SD-EF1, EF1, and EF[1,1] allocations for common preferences whenever the number of agents is at least $3\Delta-1$. We further provide, for any fixed $\eps>0$, deterministic polynomial-time algorithms that find such allocations whenever the number of agents is at least $(3+\eps)\Delta$. These guarantees strengthen earlier work by Barman and Viswanathan on equitable colorings.
\end{abstract}

\newpage

\section{Introduction}

The {\em fair allocation} problem concerns the task of partitioning a set of indivisible resources among several agents with diverse preferences. This problem has received sustained attention across a variety of research fields, including mathematics, economics, and computer science, and arises naturally in numerous real-world applications, such as allocating computing resources among users and assigning course seats to students. In the standard setting, we are given a set $M$ of resources, referred to as {\em items}, and a set of $\ell$ agents, each endowed with preferences over the subsets of $M$. An {\em allocation} of $M$ to the agents is a sequence $(A_1, \ldots, A_\ell)$ of pairwise disjoint sets whose union is $M$, where $A_i$ denotes the set of items assigned to agent $i$. Given the agents' preferences, the objective is to find an allocation that satisfies a prescribed fairness criterion.

One natural concept of fairness, introduced by Bogomolnaia and Moulin~\cite{BogomolnaiaM01} and further developed by Aziz, Gaspers, Mackenzie, and Walsh~\cite{AzizGMW15}, relies on the notion of stochastic dominance to compare assignments. In this setting, each agent $i \in [\ell]$ is associated with a weak order $\succeq_i$ on $M$, representing an ordinal ranking of the items. An allocation is called {\em stochastic-dominance envy-free} (SD-EF) if the assignment of each agent $i \in [\ell]$ stochastically dominates the assignment of any other agent under her preference order $\succeq_i$, meaning that for every item $\alpha \in M$, the agent's assignment includes at least as many items weakly preferred to $\alpha$ under $\succeq_i$ as any other agent's assignment does. However, this condition is often too demanding and is rarely achievable. This gives rise to the relaxation of this notion, called {\em stochastic-dominance envy-freeness up to one item} (SD-EF1), which allows a hypothetical removal of at most one item from the other agent's assignment before this comparison (see, e.g.,~\cite{BarmanELS25,CooksonE025}).

Another common way to model preferences in fair allocation is through valuation functions, where each agent $i \in [\ell]$ is equipped with a valuation $w_i : \calP(M) \to \R$ that assigns to every subset of $M$ a real value representing the extent to which she desires that subset. The {\em envy-freeness} requirement in this setting asks each agent to prefer her own assignment at least as much as that of any other agent, namely, for all $i, i' \in [\ell]$, we require that $w_i(A_i) \geq w_i(A_{i'})$. As before, this notion is naturally relaxed to that of {\em envy-freeness up to one item} (EF1), first proposed by Budish~\cite{Budish11}, which demands that no agent prefers the assignment of another agent over her own if we hypothetically exclude at most one item from the other agent's assignment, namely, for all $i,i' \in [\ell]$, there exists a set $S \subseteq A_{i'}$ with $|S| \leq 1$ such that $w_i(A_i) \geq w_i(A_{i'} \setminus S)$. When the valuations are non-decreasing, meaning that adding items to a set cannot decrease its value, such an allocation is known to exist and can be efficiently found~\cite{LiptonMMS04}. For additive valuations that may assign both positive and negative values to individual items, it is natural to consider the fairness notion of {\em envy-freeness up to one item from each side} (EF[1,1]), introduced by Shoshan, Hazon, and Segal-Halevi~\cite{ShoshanHS23}. Here, each agent is required to prefer her own assignment to that of any other agent, under a removal of at most one item from each assignment, namely, for all $i,i' \in [\ell]$, there exist sets $S_1 \subseteq A_i$ and $S_2 \subseteq A_{i'}$ with $|S_1| \leq 1$ and $|S_2| \leq 1$, such that $w_i(A_i \setminus S_1) \geq w_i(A_{i'} \setminus S_2)$.

In addition to fairness considerations, in various scenarios one has to take into account feasibility constraints that forbid certain allocations due to inherent limitations. Examples of such constraints include cardinality constraints~\cite{BiswasB18}, capacity constraints~\cite{ShoshanHS23}, matroid constraints~\cite{BiswasB19}, balanced assignment sizes~\cite{CaragiannisKMPS19}, and connectivity~\cite{Igarashi23} (see~\cite{Suksompong21} for a survey). In the present paper, we focus on {\em conflict constraints}, which forbid certain pairs of items from being assigned to the same agent. As an example, consider the task of allocating artworks to exhibition rooms, where certain pairs of artworks cannot be displayed in the same room because of curatorial or spatial constraints. A framework for studying such situations, proposed by Chiarelli, Krnc, Milani\v{c}, Pferschy, Piva\v{c}, and Schauer~\cite{ChiarelliKMPPS23} and by Hummel and Hetland~\cite{HummelH22}, represents the potential conflicts by a graph $G$ on the item set $M$, called the {\em conflict graph}, in which adjacent items cannot be assigned to the same agent. An agent's assignment is called {\em feasible} with respect to the conflict graph $G$ if it forms an independent set in $G$, and an allocation to $\ell$ agents is called {\em feasible} if each agent's assignment is feasible, or equivalently, if its sets are the color classes of a proper $\ell$-coloring of $G$. The question then becomes whether a given conflict graph $G$ admits a feasible allocation that satisfies a prescribed fairness property with respect to the agents' preferences.

This question was investigated by Hummel and Hetland~\cite{HummelH22} for EF1 allocations under non-decreasing additive valuations. They proved that if the number of agents is at least the size of a largest connected component in the conflict graph $G$, then there exists a feasible EF1 allocation of $G$ for every choice of non-decreasing additive valuations, and such an allocation can be found in polynomial time. On the other hand, they showed that if the number of agents is at most the maximum degree of $G$, then some common non-decreasing additive valuation for the agents precludes feasible EF1 allocations. The authors of~\cite{HummelH22} further showed that, for certain graphs, feasible EF1 allocations may fail to exist even when the number of agents exceeds the maximum degree and when the agents have a common non-decreasing additive valuation. They posed the problem of obtaining a finer characterization of the instances that admit feasible EF1 allocations.

The problem of fair allocation under conflict constraints was further explored by Equbal, Gurjar, Igarashi, Kumar, Manurangsi, Nath, Saxena, Vaish, and Yoneda~\cite{EGIKMNSVY26} (see also~\cite{KumarEGNV24,IMY25}). They considered the relaxed setting of {\em partial maximal} allocations of graphs, in which some items may remain unassigned as long as they are adjacent to an item in each agent's assignment. For two agents, they proved that every graph admits a feasible partial maximal EF1 allocation for any non-decreasing valuations, and that in the additive case, such an allocation can be found in polynomial time. They also showed that, for three or more agents, every path graph admits a (complete) feasible EF[1,1] allocation whenever the agents share a common additive valuation. Interestingly, their proof of this result relies on the celebrated `cycle plus triangles' theorem, conjectured by Erd\H{o}s~\cite{Erdos90} and proved by Fleischner and Stiebitz~\cite{FleischnerS92,FleischnerS97}. The theorem asserts that a graph on $3m$ vertices whose edge set is the union of a Hamilton cycle and $m$ pairwise vertex-disjoint triangles is $3$-colorable (see also~\cite{Sachs93}).\footnote{The formulation in~\cite{Erdos90,FleischnerS92} assumes that the Hamilton cycle and the triangles are {\em edge-disjoint}. This version implies the one stated here. Indeed, for each edge shared by the cycle and a triangle, replace it on the cycle with a path on six new internal vertices, add two triangles on alternating triples of these vertices, and keep the original edge only in its triangle. Applying the edge-disjoint version to the resulting graph gives a coloring that restricts to the desired coloring of the original graph.} More recently, Markakis and Samaris~\cite{MarkakisS26} studied feasible EF1 allocations under conflict constraints for non-decreasing additive valuations, focusing on regimes where the number of items is bounded in terms of the number of agents and the maximum degree of the conflict graph.
               
Before turning to our contribution, let us present the notion of the {\em strong chromatic number}, which was introduced in the early nineties independently by Alon~\cite{Alon92} and Fellows~\cite{Fellows90} and plays a central role in this work. The strong chromatic number of a graph $G$, denoted $\chi_s(G)$, is the least positive integer $\ell$ such that, after adding fewer than $\ell$ isolated vertices to $G$ so that the number of vertices becomes divisible by $\ell$, for every partition of the augmented vertex set into parts of size exactly $\ell$, the resulting graph admits a proper $\ell$-coloring that assigns distinct colors to the vertices within each part (see Definition~\ref{def:chi_s}). For example, the aforementioned `cycle plus triangles' theorem~\cite{FleischnerS92} can be formulated as $\chi_s(C_{3m}) \leq 3$ for every positive integer $m$, where, as usual, $C_{m}$ denotes the cycle graph on $m$ vertices. Another variant of the strong chromatic number was proposed by Gutner and Tarsi~\cite{Gutner92,GutnerT09} and has been employed in several subsequent works, e.g.,~\cite{Ohman,AharoniBZ07,GrafH20}. This variant, which we denote by $\chi_s^{\infty}(G)$, is defined similarly, except that instead of adding isolated vertices for divisibility, one considers the original graph $G$ and requires the coloring condition for every partition of its vertex set into parts of size at most $\ell$. A discussion of these related notions is provided by \"Ohman in~\cite[Section~6]{Ohman}.

A major theme in the study of the strong chromatic number concerns its relation to the maximum degree. Alon proved in~\cite{Alon92} that there exists a constant $c>0$, such that for every graph $G$ with maximum degree $\Delta \geq 1$, it holds that $\chi_s(G) \leq c \cdot \Delta$. This result was tightened by Haxell~\cite{Haxell04} to $\chi_s(G) \leq 3\Delta-1$. The `strong $2\Delta$-colorability' conjecture, explicitly stated by Aharoni, Berger, and Ziv~\cite{AharoniBZ07}, asserts that the bound can be replaced by $2\Delta$. If true, this bound would be tight, as witnessed by the complete bipartite graph with parts of size $\Delta$ (see, e.g.,~\cite{AxenovichM06}). On the algorithmic front, a result of Harris~\cite{Harris23}, building on a series of results~\cite{HarrisS17,GrafH20,GrafHH22}, states that for any fixed $\eps >0$, there exists a deterministic polynomial-time algorithm that, given a graph $G$ with maximum degree $\Delta$ and a partition of its vertex set into parts of size $\ell$, where $\ell$ is a positive integer satisfying $\ell \geq (3+\eps) \Delta$, finds a proper $\ell$-coloring of $G$ assigning distinct colors to the vertices of each part.

\subsection{Our Contribution}

The present paper studies fair allocation under conflict constraints for agents with common preferences, focusing on three fairness criteria: SD-EF1, EF1, and EF[1,1]. For a given conflict graph, we aim to determine the numbers of agents for which a feasible fair allocation is guaranteed to exist under every common preference. Our pivotal conceptual contribution lies in revealing intimate connections between these questions and variants of the strong chromatic number of the conflict graph. To this end, we introduce a hierarchy of strong chromatic numbers and show that two members of this hierarchy, which we refer to as the strong chromatic numbers of types~$1$ and~$2$, are closely tied to the fairness notions studied here. This viewpoint enables us to isolate the algorithmic core of the allocation problems and thereby obtain efficient allocation algorithms in several settings. We now turn to a detailed description of our contribution, summarized in Table~\ref{tab:summary_results}.

We begin by introducing our hierarchy of strong chromatic numbers. For a positive integer~$r$, the {\em strong chromatic number of type $r$} of a graph $G$, denoted $\chi^r_s(G)$, is the least positive integer $\ell$ such that for every partition of the vertex set of $G$ into parts of size at most $\ell$, with at most $r$ parts of size strictly smaller than $\ell$, there exists a proper $\ell$-coloring of $G$ that assigns distinct colors to the vertices of each part (see Definition~\ref{def:chi_s^r}). Note that when the limitation on the number of parts of size smaller than $\ell$ is dropped, we recover the quantity $\chi_s^\infty(G)$ studied in~\cite{Gutner92,GutnerT09}. Comparing these quantities with the strong chromatic number $\chi_s(G)$ of~\cite{Alon92,Fellows90}, it turns out that every graph $G$ satisfies
\begin{eqnarray}\label{eq:chain}
\chi_s^1(G) \leq \chi_s^2(G) \leq \chi_s(G) \leq \chi_s^\infty(G).
\end{eqnarray}
For a detailed study of the hierarchy, including a proof of this chain of inequalities, see Section~\ref{sec:chi_s^r}.

Our first main result concerns the fairness criterion of stochastic-dominance envy-freeness up to one item (SD-EF1). We provide a full characterization of the numbers of agents that guarantee a feasible SD-EF1 allocation with respect to every common weak order. The following theorem shows that the strong chromatic number of type $1$ of the conflict graph is the precise threshold.

\begin{theorem}\label{thm:intro_SD}
Let $G=(V,E)$ be a graph, and let $\ell$ be a positive integer.
\begin{enumerate}
  \item\label{itm:1} If $\ell \geq \chi_s^1(G)$, then for every weak order $\succeq$ on $V$, the graph $G$ admits a feasible SD-EF1 allocation to $\ell$ agents with respect to the common weak order $\succeq$.
  \item\label{itm:2} If $\ell < \chi_s^1(G)$, then there exists a weak order $\succeq$ on $V$, such that the graph $G$ admits no feasible SD-EF1 allocation to $\ell$ agents with respect to the common weak order $\succeq$.
\end{enumerate}
\end{theorem}

The characterization given in Theorem~\ref{thm:intro_SD} crucially relies on the assumption that all agents share a common weak order. For heterogeneous weak orders, the strong chromatic number of type $1$ no longer provides a sufficient condition for the existence of feasible SD-EF1 allocations. This already occurs for graphs with this quantity equal to two, as follows from constructions obtained in somewhat different settings by Barman, Ebadian, Latifian, and Shah~\cite{BarmanELS25} and by Cookson, Ebadian, and Shah~\cite{CooksonE025}. For details, see Proposition~\ref{prop:3K2+K1}.

We next consider the fairness notion of envy-freeness up to one item (EF1) for agents sharing a common monotone additive valuation, where by monotone we mean either non-decreasing or non-increasing. As a consequence of Item~\ref{itm:1} of Theorem~\ref{thm:intro_SD}, we obtain that the strong chromatic number of type $1$ again yields a sufficient condition for the existence of feasible fair allocations.

\begin{theorem}\label{thm:intro_EF1_upper}
Let $G=(V,E)$ be a graph, and let $\ell$ be an integer satisfying $\ell \geq \chi_s^1(G)$. Then, for every monotone additive valuation $w:\calP(V) \to \R$, the graph $G$ admits a feasible EF1 allocation to $\ell$ agents with respect to the common valuation $w$.
\end{theorem}

For arbitrary additive valuations, we consider the fairness notion of envy-freeness up to one item from each side (EF[1,1]). In this setting, the strong chromatic number of type $2$ provides a sufficient condition for the existence of feasible fair allocations.

\begin{theorem}\label{thm:intro_EF11_upper}
Let $G=(V,E)$ be a graph, and let $\ell$ be an integer satisfying $\ell \geq \chi_s^2(G)$. Then, for every additive valuation $w:\calP(V) \to \R$, the graph $G$ admits a feasible EF[1,1] allocation to $\ell$ agents with respect to the common valuation $w$.
\end{theorem}

For monotone additive valuations, the notions of EF1 and EF[1,1] are equivalent. Therefore, Theorem~\ref{thm:intro_EF11_upper} also yields a sufficient condition for the existence of feasible EF1 allocations for such valuations. However, there exist graphs for which the strong chromatic number of type $1$ is strictly smaller than that of type $2$. For example, for every integer $m \geq 2$, the complete bipartite graph $K_{m,m}$ with $m$ vertices in each part satisfies $\chi_s^1(K_{m,m}) = \lceil 3m/2 \rceil$, whereas $\chi_s^2(K_{m,m}) = 2m$ (see Lemma~\ref{lemma:K_m,m}). For such graphs, Theorem~\ref{thm:intro_EF1_upper} provides a tighter condition in the monotone setting. On the other hand, for general additive valuations, the dependence on the strong chromatic number of type $2$ in Theorem~\ref{thm:intro_EF11_upper} cannot be weakened to the strong chromatic number of type $1$. This is again demonstrated by the family of complete bipartite graphs (see Proposition~\ref{prop:Km,m}).

Theorem~\ref{thm:intro_EF11_upper} extends and strengthens earlier results in several directions. First, it significantly extends the aforementioned result of~\cite{EGIKMNSVY26}, which establishes the corresponding guarantee for path graphs, a graph family with strong chromatic number of type $2$ bounded by three (see Lemma~\ref{lemma:chi_s^2_cycles}). It also strengthens, in the common valuation setting, the sufficient condition for feasible EF1 allocations given in~\cite{HummelH22}. The latter requires the number of agents to be at least the maximum size of a connected component in the graph, a quantity that forms an upper bound on the strong chromatic number (of any type) and may be arbitrarily larger (see Lemma~\ref{lemma:chi_s^inf<=C}, derived from~\cite{Ohman}). Moreover, the result of~\cite{HummelH22} applies to non-decreasing additive valuations, whereas Theorem~\ref{thm:intro_EF11_upper} applies to arbitrary additive valuations.

The sufficient conditions for feasible fair allocations, given in Theorems~\ref{thm:intro_SD},~\ref{thm:intro_EF1_upper}, and~\ref{thm:intro_EF11_upper}, imply such conditions in terms of the maximum degree of the conflict graph. Indeed, by combining them with Haxell's bound~\cite{Haxell04} on the strong chromatic number and with the inequalities in~\eqref{eq:chain}, it follows that every graph $G$ with maximum degree $\Delta$ admits feasible SD-EF1, EF1, and EF[1,1] allocations to $\ell$ agents with respect to common preferences of the relevant kind, provided that $\ell \geq 3\Delta-1$. Moreover, the `strong $2\Delta$-colorability' conjecture~\cite{AharoniBZ07}, if true, would yield the sharper sufficient condition $\ell \geq 2\Delta$. These bounds are within a constant multiplicative factor of what is necessary in general, since for $\ell \leq \Delta$, there is an appropriate common preference for which $G$ admits no feasible fair allocation to $\ell$ agents. For SD-EF1, this follows from Item~\ref{itm:2} of Theorem~\ref{thm:intro_SD}, combined with the fact that $\chi_s^1(G) \geq \Delta+1$ (see Lemma~\ref{lemma:chi_s_1-D}). For EF1 and EF[1,1], this follows from the non-existence result of~\cite{HummelH22}.

In light of Theorems~\ref{thm:intro_EF1_upper} and~\ref{thm:intro_EF11_upper}, it is natural to ask whether the stated sufficient conditions on the number of agents, in terms of the strong chromatic numbers of types $1$ and $2$, respectively, are necessary for guaranteeing feasible fair allocations. For graphs whose relevant strong chromatic number is at most three, we show that the answer is positive, in the sense that fewer agents do not suffice for a feasible EF1 allocation with respect to a certain common non-decreasing additive valuation, and therefore do not suffice for EF[1,1] either (see Proposition~\ref{prop:chi_s^1=3}). However, already at value four, we show that the sufficient conditions are not necessary in general. This highlights a significant difference from the SD-EF1 setting.

\begin{theorem}\label{thm:intro_EF1_lower}
There exists a graph $G=(V,E)$ with $\chi_s^1(G) = 4$, such that for every monotone additive valuation $w:\calP(V) \to \R$, the graph $G$ admits a feasible EF1 allocation to three agents with respect to the common valuation $w$.
\end{theorem}

\begin{theorem}\label{thm:intro_EF11_lower}
There exists a graph $G=(V,E)$ with $\chi_s^2(G) = 4$, such that for every additive valuation $w:\calP(V) \to \R$, the graph $G$ admits a feasible EF[1,1] allocation to three agents with respect to the common valuation $w$.
\end{theorem}

To pave the way for our algorithmic results, we now describe the key connection between fair allocation and strong colorability, which underlies all our existential results and is inspired by an argument of~\cite{EGIKMNSVY26} for path graphs. We first explain this connection for the SD-EF1 fairness notion. Given a conflict graph $G$ and $\ell$ agents equipped with a common weak order, we arrange the vertices of $G$ according to this weak order, from most preferred to least preferred, and split the resulting list into consecutive blocks of size $\ell$, except possibly for the last block which may be smaller. It can be shown that every allocation that assigns one vertex from each full block to each agent, and assigns the vertices of the last block to distinct agents, is SD-EF1 (see Lemma~\ref{lemma:SD_Pj}). We therefore seek such an allocation that is also feasible, meaning that each agent receives an independent set of $G$. Such an allocation naturally corresponds to a proper $\ell$-coloring of $G$ whose color classes intersect each block in at most one vertex. The existence of such a coloring is guaranteed by the assumption $\ell \geq \chi_s^1(G)$, together with an appropriate monotonicity property of strong colorability (see Lemma~\ref{lemma:monotone_chi_r}). This, in turn, yields a feasible SD-EF1 allocation of $G$ to $\ell$ agents with respect to the given common weak order. For EF1 fairness, the allocation is obtained similarly, with the vertices arranged according to the weak order induced by the given common monotone additive valuation (see Lemma~\ref{lemma:EF1_Pj}). For EF[1,1], when the additive valuation is not necessarily monotone, we apply the same idea separately to vertices with non-negative values and to vertices with non-positive values (see Lemma~\ref{lemma:EF11_P_Q}). This produces at most two blocks of size strictly smaller than $\ell$, which leads to the stronger condition $\ell \geq \chi_s^2(G)$ used for EF[1,1].

The existential arguments described above reduce the algorithmic allocation problems to the following coloring task. The input consists of a graph $G$, a positive integer $\ell$, and a partition of the vertex set of $G$ into parts of size at most $\ell$, with at most one or two parts of size strictly smaller than $\ell$, depending on the fairness notion, and the goal is to find a proper $\ell$-coloring of $G$ that assigns distinct colors to the vertices of each part. Any efficient algorithm for this coloring problem yields efficient algorithms for the associated fair allocation problems (see Theorem~\ref{thm:algo_gen} for a formal statement). Although such a coloring is guaranteed to exist whenever $\ell$ is at least the relevant strong chromatic number, no efficient algorithm is known for finding one in general. In fact, even the algorithmic counterpart of the `cycle plus triangles' theorem~\cite{FleischnerS92}, for cycles $C_{3m}$ and parts of size three, is a long-standing open problem (see, e.g.,~\cite{FleischnerS97,AlonChallenges02,Haviv22-FISC}). Nonetheless, by combining our approach with an algorithm of Harris~\cite{Harris23}, we obtain the following algorithmic results in terms of the maximum degree, presented separately for each fairness notion.

\begin{theorem}\label{thm:intro_algos}
For any fixed $\eps >0$, there exists a deterministic polynomial-time algorithm for each of the following problems.
\begin{enumerate}
  \item\label{itm:algSD} Given a graph $G$ with maximum degree $\Delta$, a positive integer $\ell \geq (3+\eps)\Delta$, and a weak order $\succeq$ on the vertex set of $G$, find a feasible SD-EF1 allocation of $G$ to $\ell$ agents with respect to the common weak order $\succeq$.
  \item\label{itm:algEF1} Given a graph $G$ with maximum degree $\Delta$, a positive integer $\ell \geq (3+\eps)\Delta$, and a monotone additive valuation $w$ on the vertex set of $G$, find a feasible EF1 allocation of $G$ to $\ell$ agents with respect to the common valuation $w$.
  \item\label{itm:algEF11} Given a graph $G$ with maximum degree $\Delta$, a positive integer $\ell \geq (3+\eps)\Delta$, and an additive valuation $w$ on the vertex set of $G$, find a feasible EF[1,1] allocation of $G$ to $\ell$ agents with respect to the common valuation $w$.
\end{enumerate}
\end{theorem}

For particular graph families, our framework can sometimes yield efficient allocation algorithms with improved thresholds on the number of agents. We demonstrate this for the family of path graphs, addressed in~\cite{EGIKMNSVY26}, where the existential results guarantee feasible fair allocations already when the number of agents is at least three. Borrowing an algorithmic result of Fleischner and Stiebitz~\cite{FleischnerS97}, we derive that such allocations can be found efficiently when the number of agents is at least four, one agent above the existential threshold (see Theorem~\ref{thm:paths_algo}). This result generalizes an algorithm of~\cite{KumarEGNV24} that applies to path graphs with agents sharing a common two-valued monotone additive valuation.

We finally mention that our EF[1,1] results can also be formulated in terms of {\em equitable colorings} of vertex-weighted graphs. In this context, given a graph and real weights assigned to its vertices, one seeks a proper coloring whose color classes have balanced total weights, up to the removal of at most one vertex from each of the two compared color classes. This is precisely a feasible EF[1,1] allocation for a common additive valuation, upon identifying color classes with agents' assignments. Barman and Viswanathan~\cite{BarmanV26} studied this problem for graphs with non-negative vertex weights, also considering an approximate notion of equitability, and obtained existential and algorithmic guarantees in terms of the maximum degree. They asked whether exact equitability up to one vertex is achievable whenever the number of colors is at least a constant multiple of the maximum degree. Our results settle this question affirmatively, showing that every graph with maximum degree $\Delta$ admits a proper $\ell$-coloring satisfying the required balance property for every positive integer $\ell \geq 3\Delta-1$. Moreover, for any fixed $\eps > 0$, such a coloring can be found in deterministic polynomial time whenever $\ell \geq (3+\eps)\Delta$.

\begin{table}[t]
\centering
\renewcommand{\arraystretch}{1.25}
\begin{tabular}{@{}p{2.0cm}@{\hspace{0.45cm}} p{3.55cm} p{3.65cm} p{3.45cm}@{}}
\toprule
Fairness & Sufficiency & Necessity & Efficient algorithm \\
\midrule

SD-EF1
& $\ell \geq \chi_s^1(G)$ \newline
Theorem~\ref{thm:intro_SD}, Item~\ref{itm:1}
& Yes \newline
Theorem~\ref{thm:intro_SD}, Item~\ref{itm:2}
& $\ell \geq (3+\eps)\Delta$ \newline
Theorem~\ref{thm:intro_algos}, Item~\ref{itm:algSD} \\

EF1
& $\ell \geq \chi_s^1(G)$ \newline
Theorem~\ref{thm:intro_EF1_upper}
& Not in general \newline
Theorem~\ref{thm:intro_EF1_lower}
& $\ell \geq (3+\eps)\Delta$ \newline
Theorem~\ref{thm:intro_algos}, Item~\ref{itm:algEF1} \\

EF[1,1]
& $\ell \geq \chi_s^2(G)$ \newline
Theorem~\ref{thm:intro_EF11_upper}
& Not in general \newline
Theorem~\ref{thm:intro_EF11_lower}
& $\ell \geq (3+\eps)\Delta$ \newline
Theorem~\ref{thm:intro_algos}, Item~\ref{itm:algEF11} \\

\bottomrule
\end{tabular}
\caption{Summary of sufficiency, necessity, and algorithmic results for feasible fair allocations of a conflict graph $G$ with maximum degree $\Delta$ to $\ell$ agents with common preferences.}
\label{tab:summary_results}
\end{table}

\subsection{Outline}
The rest of this paper is organized as follows. In Section~\ref{sec:preliminaries}, we gather definitions and facts on fair allocation. In Section~\ref{sec:chi_s^r}, we introduce our hierarchy of strong chromatic numbers and establish some useful properties. In Section~\ref{sec:SD}, we study feasible SD-EF1 allocations for common weak orders and establish the characterization stated in Theorem~\ref{thm:intro_SD}. In Section~\ref{sec:EF}, we study feasible EF1 and EF[1,1] allocations for common additive valuations and establish Theorems~\ref{thm:intro_EF1_upper},~\ref{thm:intro_EF11_upper},~\ref{thm:intro_EF1_lower}, and~\ref{thm:intro_EF11_lower}. In Section~\ref{sec:algorithms}, we present our algorithmic applications, including Theorem~\ref{thm:intro_algos}. We end with Section~\ref{sec:remarks}, where we present concluding remarks and discuss avenues for further research.

\section{Preliminaries}\label{sec:preliminaries}

\subsection{Notation}

For a positive integer $n$, let $[n]$ denote the set of positive integers up to $n$. For a set $M$, let $\calP(M)$ denote the power set of $M$. A {\em partition} of $M$ is a collection of non-empty pairwise disjoint subsets of $M$, whose union is $M$.

\subsection{Fair Allocation}

Let $M$ be a finite set of indivisible items, and let $\ell$ be a positive integer. An {\em allocation of $M$ to $\ell$ agents} is an ordered sequence $(A_1, \ldots, A_\ell)$ of $\ell$ pairwise disjoint subsets of $M$, possibly empty, whose union is $M$. For each $i \in [\ell]$, $A_i$ is the set of items allocated to agent $i$.

We use two preference models. In the first, each agent is associated with a weak order on $M$, that is, a complete and transitive binary relation on $M$, possibly with ties. The corresponding fairness notion is stochastic-dominance envy-freeness up to one item, abbreviated SD-EF1. It is defined as follows.

\begin{definition}
Let $M$ be a finite set of items, and let $\succeq$ be a weak order on $M$. For two sets $A,B \subseteq M$, we say that $A$ {\em stochastically dominates} $B$ with respect to $\succeq$, and write $A \succeq^{\SD} B$, if for every $\alpha \in M$, the set $A$ includes at least as many items weakly preferred to $\alpha$ under $\succeq$ as $B$ does, namely,
\[ | A \cap S_\alpha| \geq |B \cap S_\alpha|,\]
where $S_\alpha = \{ x \in M \mid x \succeq \alpha\}$.
For a positive integer $\ell$, an allocation $(A_1, \ldots, A_\ell)$ of $M$ to $\ell$ agents is called {\em stochastic-dominance envy-free up to one item} (SD-EF1) with respect to weak orders $\succeq_1,\ldots, \succeq_\ell$ on $M$ if for all $i,i' \in [\ell]$, there exists a set $S \subseteq A_{i'}$ with $|S| \leq 1$, such that $A_i \succeq_i^{\SD} A_{i'} \setminus S$. If, in addition, the $\ell$ weak orders are all equal to a single weak order $\succeq$, we say that the allocation is SD-EF1 with respect to the common weak order $\succeq$.
\end{definition}

In the second preference model, each agent is associated with a valuation function on the subsets of $M$. We consider the fairness notion of envy-freeness up to one item from each side, abbreviated EF[1,1], defined as follows.

\begin{definition}\label{def:EF11}
Let $M$ be a finite set of items.
For a positive integer $\ell$, an allocation $(A_1, \ldots, A_\ell)$ of $M$ to $\ell$ agents is called {\em envy-free up to one item from each side} (EF[1,1]) with respect to valuations $w_1,\ldots, w_\ell: \calP(M) \to \R$ if for all $i,i' \in [\ell]$, there exist sets $S_1 \subseteq A_i$ and $S_2 \subseteq A_{i'}$ with $|S_1| \leq 1$ and $|S_2| \leq 1$, such that $w_i(A_i \setminus S_1) \geq w_i(A_{i'} \setminus S_2)$.
If, in addition, the $\ell$ valuations are all equal to a single valuation $w$, we say that the allocation is EF[1,1] with respect to the common valuation $w$.
\end{definition}

A valuation $w: \calP(M) \to \R$ is called {\em non-decreasing} if $w(A) \leq w(B)$ whenever $A \subseteq B \subseteq M$ and {\em non-increasing} if $w(A) \geq w(B)$ whenever $A \subseteq B \subseteq M$. It is called {\em monotone} if it is either non-decreasing or non-increasing. In both monotone cases, we refer to the corresponding special case of EF[1,1] from Definition~\ref{def:EF11} as EF1. In the former, the set $S_1$ in the definition can always be taken to be empty, and in the latter, the set $S_2$ can always be taken to be empty. We will be particularly interested in {\em additive} valuations $w: \calP(M) \to \R$, for which $w(A) = \sum_{x \in A}{w(x)}$ holds for every set $A \subseteq M$, where, by abuse of notation, $w(x) = w(\{x\})$. In particular, $w(\emptyset)=0$. In the additive case, an item $x \in M$ with $w(x) \geq 0$ is called a {\em good}, and an item $x \in M$ with $w(x) \leq 0$ is called a {\em chore}. Thus, an item with value $0$ may be treated as either a good or a chore. Note that an additive valuation is non-decreasing if all items are goods, and is non-increasing if all items are chores.

We will use the following standard relation between SD-EF1 and EF1.
Here and throughout, for an additive valuation $w: \calP(M) \to \R$, the {\em weak order $\succeq$ induced by $w$} is defined by setting, for all $x,y \in M$, that $x \succeq y$ if and only if $w(x) \geq w(y)$.

\begin{lemma}\label{lemma:SD->EF1}
Let $M$ be a finite set of items, let $\ell$ be a positive integer, let $w_1, \ldots, w_\ell: \calP(M) \to \R$ be non-decreasing additive valuations, and let $\succeq_1, \ldots, \succeq_\ell$ be the weak orders induced by $w_1, \ldots, w_\ell$, respectively. If an allocation of $M$ to $\ell$ agents is SD-EF1 with respect to $\succeq_1, \ldots, \succeq_\ell$, then it is also EF1 with respect to $w_1, \ldots, w_\ell$.
\end{lemma}

\begin{proof}
It is sufficient to prove that for every non-decreasing additive valuation $w: \calP(M) \to \R$ and for the weak order $\succeq$ it induces, if two sets $A,B \subseteq M$ satisfy $A \succeq^{\SD} B$, then $w(A) \geq w(B)$.
Suppose that $A \succeq^{\SD} B$, and let $\lambda_1 > \cdots > \lambda_r$ denote the distinct values under $w$ of the items in $A \cup B$. Since $w$ is non-decreasing and additive, every item has a non-negative value, hence $\lambda_r \geq 0$. Set $\lambda_{r+1}=0$.
For every $j \in [r]$, let $A_j = \{x \in A \mid w(x) \geq \lambda_j\}$ and $B_j = \{x \in B \mid w(x) \geq \lambda_j\}$. Since $A \succeq^{\SD} B$, it follows that $|A_j| \geq |B_j|$ for all $j \in [r]$. Observe that $w(A) = \sum_{j=1}^{r}{(\lambda_j-\lambda_{j+1})|A_j|}$ and $w(B) = \sum_{j=1}^{r}{(\lambda_j-\lambda_{j+1})|B_j|}$. Since $\lambda_j-\lambda_{j+1} \geq 0$ for all $j \in [r]$, it follows that $w(A) \geq w(B)$. This completes the proof.
\end{proof}

The following lemma shows that separate EF1 allocations for goods and chores can be combined to an EF[1,1] allocation.

\begin{lemma}\label{lemma:EF1->EF11}
Let $M$ be a finite set of items, let $w: \calP(M) \to \R$ be an additive valuation, and consider a partition of $M$ into two sets $M^+$ and $M^-$, consisting of goods and chores, respectively, under $w$. For a positive integer $\ell$, let $\calB = (B_1, \ldots, B_\ell)$ be an EF1 allocation of $M^+$ to $\ell$ agents with respect to the restriction of $w$ to $M^+$, and let $\calC = (C_1, \ldots, C_\ell)$ be an EF1 allocation of $M^-$ to $\ell$ agents with respect to the restriction of $w$ to $M^-$. Then $(B_1 \cup C_1, \ldots, B_\ell \cup C_\ell)$ is an EF[1,1] allocation of $M$ to $\ell$ agents with respect to the common valuation $w$.
\end{lemma}

\begin{proof}
Let $i,i' \in [\ell]$. Since $\calB$ is EF1 with respect to the non-decreasing restriction of $w$ to $M^+$, there exists a set $S^+ \subseteq B_{i'}$ with $|S^+| \leq 1$, such that $w(B_i) \geq w(B_{i'} \setminus S^+)$. Since $\calC$ is EF1 with respect to the non-increasing restriction of $w$ to $M^-$, there exists a set $S^- \subseteq C_{i}$ with $|S^-| \leq 1$, such that $w(C_i \setminus S^-) \geq w(C_{i'})$. By the additivity of $w$, it follows that
\[w( (B_i \cup C_i) \setminus S^-) = w(B_i)+ w(C_i \setminus S^-) \geq w(B_{i'} \setminus S^+) + w(C_{i'}) = w( (B_{i'} \cup C_{i'}) \setminus S^+),\]
hence the EF[1,1] property is satisfied.
\end{proof}

\subsection{Conflict Constraints}

We recall some standard terminology from graph theory. Throughout the paper, all graphs are finite and simple. A set of vertices in a graph $G$ is called {\em independent} if no two of its vertices are adjacent. For a positive integer $\ell$, an {\em $\ell$-coloring of $G$} is an assignment of one of $\ell$ colors to each vertex in $G$. The coloring is called {\em proper} if the endpoints of each edge in $G$ are assigned distinct colors. The set of vertices assigned a given color is called a {\em color class}. A graph $G$ is called {\em $\ell$-colorable} if it admits a proper $\ell$-coloring, and the {\em chromatic number} of $G$, denoted $\chi(G)$, is the least positive integer $\ell$ for which $G$ is $\ell$-colorable.

The fair allocation problem under conflict constraints, proposed in~\cite{ChiarelliKMPPS23,HummelH22}, is defined as follows. Let $G=(V,E)$ be a graph, whose vertices represent items. We refer to $G$ as a {\em conflict graph}. For a positive integer $\ell$, an allocation of $G$ to $\ell$ agents is an allocation of its vertex set $V$ to $\ell$ agents, namely, a sequence $(A_1, \ldots, A_\ell)$ of pairwise disjoint sets whose union is $V$. The allocation is called {\em feasible} if $A_i$ is an independent set in $G$ for every $i \in [\ell]$, equivalently, the sets $A_1, \ldots, A_\ell$ are the color classes of a proper $\ell$-coloring of $G$. We stress that all allocations in this paper are required to cover the entire vertex set, unlike the partial allocations considered, e.g., in~\cite{ChiarelliKMPPS23,EGIKMNSVY26}. An instance of the fair allocation problem under conflict constraints consists of a conflict graph $G$, a positive integer $\ell$, and preferences of the $\ell$ agents, given either by weak orders or by valuations. The question is whether such an instance admits a feasible allocation satisfying a desired fairness property, such as SD-EF1, EF1, or EF[1,1].

\section{A Hierarchy of Strong Chromatic Numbers}\label{sec:chi_s^r}

In this section, we introduce a hierarchy of strong chromatic numbers and establish several fundamental properties that will be used later. We first recall the strong chromatic number due to Alon~\cite{Alon92} and Fellows~\cite{Fellows90}.

\begin{definition}\label{def:chi_s}
Let $G$ be a graph on $n$ vertices. For a positive integer $\ell$ that divides $n$, the graph $G$ is called {\em strongly $\ell$-colorable} if for every partition of the vertex set into parts of size exactly $\ell$, there exists a proper $\ell$-coloring of $G$ such that each color class intersects each part in exactly one vertex. For a positive integer $\ell$ that does not divide $n$, the graph $G$ is called strongly $\ell$-colorable if the graph obtained from $G$ by adding to it $\ell \cdot \lceil n/\ell \rceil - n$ isolated vertices is strongly $\ell$-colorable in the sense defined above. The {\em strong chromatic number} of $G$, denoted $\chi_s(G)$, is the least positive integer $\ell$ for which $G$ is strongly $\ell$-colorable.
\end{definition}

We now introduce the hierarchy of strong chromatic numbers.

\begin{definition}\label{def:chi_s^r}
Let $G$ be a graph. For positive integers $\ell$ and $r$, the graph $G$ is called {\em strongly $(\ell,r)$-colorable} if for every partition of the vertex set into parts of size at most $\ell$, with at most $r$ parts of size strictly smaller than $\ell$, there exists a proper $\ell$-coloring of $G$ such that each color class intersects each part in at most one vertex. The {\em strong chromatic number of type $r$} of $G$, denoted $\chi^r_s(G)$, is the least positive integer $\ell$ for which $G$ is strongly $(\ell,r)$-colorable.
\end{definition}
\noindent
By abuse of notation, we also allow taking $r = \infty$ in Definition~\ref{def:chi_s^r}. In this case, the number of parts of size smaller than $\ell$ is unbounded, and the resulting quantity $\chi_s^\infty(G)$ coincides with the one studied in~\cite{Gutner92,GutnerT09}. Note that, for every graph $G$, the value of $\chi_s^r(G)$ is non-decreasing in $r$.

We remark that the notions of strong colorability above differ from each other in the collection of partitions for which one requires a suitable proper coloring. For strong $(\ell,r)$-colorability, this collection consists of all partitions of the vertex set into parts of size at most $\ell$ with at most $r$ parts of size strictly smaller than $\ell$. The notion of strong $\ell$-colorability from Definition~\ref{def:chi_s} can also be phrased in this language by considering all partitions of the vertex set into $\lceil n/\ell \rceil$ parts of size at most $\ell$, where $n$ denotes the number of vertices.

The following lemma provides a chain of inequalities involving some of the quantities discussed above.
\begin{lemma}\label{lemma:compare_chi_s}
For every graph $G$, it holds that
\[ \chi(G) \leq \chi_s^1 (G) \leq \chi_s^2 (G) \leq \chi_s (G) \leq \chi_s^\infty (G).\]
\end{lemma}

\begin{proof}
All inequalities except the third follow directly from the definitions.
For the third, it suffices to show that if a graph is strongly $\ell$-colorable, then it is also strongly $(\ell,2)$-colorable. Assume that a graph $G=(V,E)$ is strongly $\ell$-colorable. Consider a partition of $V$ into $t$ parts $Q_1, \ldots, Q_t$ of size at most $\ell$, with at most two parts of size strictly smaller than $\ell$. Let $G'=(V',E')$ be the graph obtained from $G$ by adding to it $\ell \cdot \lceil |V|/\ell \rceil -|V|$ isolated vertices. We define below a partition of $V'$ into parts of size exactly $\ell$, such that each of the original parts $Q_1, \ldots, Q_t$ is contained in a part of this new partition. Then, since $G$ is strongly $\ell$-colorable, the graph $G'$ admits a proper $\ell$-coloring such that each color class intersects each part of the new partition in exactly one vertex. Restricting this coloring to $V$ gives a proper $\ell$-coloring of $G$ such that each color class intersects each of the parts $Q_1, \ldots, Q_t$ in at most one vertex, as required.

To obtain the required partition of $V'$, we use the added isolated vertices to complete the deficient parts to size $\ell$, as described next. If no part among $Q_1, \ldots, Q_t$ has size smaller than $\ell$, then $V=V'$, so the given partition satisfies the desired property. If exactly one of the parts has size smaller than $\ell$, say $Q_t$, then the partition is obtained from $Q_1, \ldots, Q_t$ by adding all vertices of $V' \setminus V$ to $Q_t$, which completes it to size $\ell$. Finally, if two of the parts have size strictly smaller than $\ell$, say $Q_{t-1}$ and $Q_t$, then we consider two cases. If $|Q_{t-1}|+|Q_t| \leq \ell$, then $Q_{t-1} \cup Q_t \cup (V' \setminus V)$ has size exactly $\ell$, and we replace the two parts by this single part. Otherwise, the number of vertices in $V' \setminus V$ is exactly $(\ell-|Q_{t-1}|)+(\ell-|Q_t|)$, so we can split them between $Q_{t-1}$ and $Q_t$ to complete both parts to size exactly $\ell$. This completes the proof.
\end{proof}

A useful property of strong colorability is monotonicity, saying that if a graph is strongly colorable with $\ell$ colors, under one of the variants discussed above, it remains so with $\ell+1$ colors. The monotonicity of strong $\ell$-colorability and of strong $(\ell,\infty)$-colorability was proved in~\cite{Ohman} and in~\cite{Gutner92,GutnerT09}, respectively. We prove below the monotonicity of strong $(\ell,r)$-colorability, for every positive integer $r$, using a similar argument.

\begin{lemma}\label{lemma:monotone_chi_r}
Let $G$ be a graph, and let $\ell$ and $r$ be positive integers.
If $G$ is strongly $(\ell,r)$-colorable, then $G$ is also strongly $(\ell+1,r)$-colorable.
\end{lemma}

\begin{proof}
Let $G=(V,E)$ be a strongly $(\ell,r)$-colorable graph. We first observe that every induced subgraph of $G$ is also strongly $(\ell,r)$-colorable. To see this, let $G'=(V',E')$ be an induced subgraph of $G$, and consider a partition of $V'$ into parts of size at most $\ell$, with at most $r$ parts of size strictly smaller than $\ell$. Extend this partition to a partition of $V$ by adding the vertices of $V \setminus V'$, one by one, as follows. For each such vertex, if there currently exists a part of size smaller than $\ell$, then add the vertex to such a part, and otherwise, initiate a new part consisting of this vertex. At the end of this process, we obtain a partition of $V$ into parts of size at most $\ell$, with at most $r$ parts of size strictly smaller than $\ell$. In fact, if at some point a new part is initiated, then there is at most one part of size smaller than $\ell$. Since $G$ is strongly $(\ell,r)$-colorable, there exists a proper $\ell$-coloring of $G$ such that each color class intersects each part in at most one vertex. Since each part of the original partition of $V'$ is contained in some part of the produced partition of $V$, and since $G'$ is an induced subgraph of $G$, the restriction of this $\ell$-coloring to $V'$ satisfies the desired property with respect to $G'$ and the given partition of $V'$.

We now show that $G$ is strongly $(\ell+1,r)$-colorable. Consider a partition of $V$ into $t$ parts $Q_1, \ldots, Q_t$ of size at most $\ell+1$, with at most $r$ parts of size strictly smaller than $\ell+1$. For each $j \in [t]$, let $Q'_j$ be an arbitrary subset of $Q_j$ of size $\ell$ if $|Q_j|=\ell+1$, and let $Q'_j = Q_j$ otherwise. Note that $|Q'_j| \leq \ell$ for all $j \in [t]$, and that at most $r$ of the sets $Q'_1, \ldots, Q'_t$ have size strictly smaller than $\ell$. Let $G'$ be the subgraph of $G$ induced by $\cup_{j \in [t]}{Q'_j}$. By the observation above, since $G$ is strongly $(\ell,r)$-colorable, so is $G'$. Therefore, there exists a proper $\ell$-coloring of $G'$ such that each color class intersects each of $Q'_1, \ldots, Q'_t$ in at most one vertex. Let $C$ be an arbitrary color class of this coloring. Since $G'$ is an induced subgraph of $G$, the set $C$ is independent in $G$. For each $j \in [t]$, if $|Q'_j| = \ell$, then $Q'_j$ includes exactly one vertex of each color. Thus, the color class $C$ intersects each such $Q'_j$, and hence each $Q_j$ of size $\ell+1$, in exactly one vertex. Consider the induced subgraph of $G$ on $V \setminus C$ and the partition of its vertex set into the parts $Q_1 \setminus C, \ldots, Q_t \setminus C$, excluding empty sets if any. All these parts have size at most $\ell$, and at most $r$ of them have size strictly smaller than $\ell$. Therefore, by applying the observation above again, this subgraph admits a proper $\ell$-coloring such that each color class intersects each of $Q_1 \setminus C, \ldots, Q_t \setminus C$ in at most one vertex. Extending this coloring to $G$ by assigning one additional color to all vertices of $C$, we obtain a proper $(\ell+1)$-coloring of $G$ such that each color class intersects each of $Q_1, \ldots, Q_t$ in at most one vertex, as desired.
\end{proof}

The following lemma, derived from~\cite{Ohman}, relates the strong chromatic number of type infinity to the sizes of the connected components.

\begin{lemma}[{\cite{Ohman}}]\label{lemma:chi_s^inf<=C}
For a graph $G$, let $\ell$ denote the largest number of vertices in a connected component of $G$. Then $\chi_s^\infty(G) \leq \ell$.
\end{lemma}

The maximum degree provides useful bounds for the hierarchy of strong chromatic numbers. The following simple lower bound already holds for the strong chromatic number of type $1$.

\begin{lemma}\label{lemma:chi_s_1-D}
For every graph $G$ with maximum degree $\Delta$, it holds that $\chi_s^1(G) \geq \Delta+1$.
\end{lemma}

\begin{proof}
For $\Delta = 0$, the claim is trivial. Assuming that $\Delta \geq 1$, by Lemma~\ref{lemma:monotone_chi_r}, it suffices to show that $G$ is not strongly $(\Delta,1)$-colorable.
To this end, consider a partition of the vertex set of $G$ with one part defined as the neighborhood of a vertex of degree $\Delta$ and with the remaining vertices partitioned arbitrarily into parts of size exactly $\Delta$, possibly except one part of smaller size. Observe that $G$ has no proper $\Delta$-coloring in which each color class intersects each part in at most one vertex. Indeed, such a coloring would have to assign distinct colors to the $\Delta$ vertices of the first part and a different color to their common neighbor. This completes the proof.
\end{proof}

For an upper bound on the strong chromatic number in terms of the maximum degree, we state a result of Haxell~\cite{Haxell04}.
\begin{theorem}[\cite{Haxell04}]\label{thm:3Delta-1}
For every graph $G$ with maximum degree $\Delta \geq 1$, it holds that $\chi_s(G) \leq 3\Delta-1$.
\end{theorem}
\noindent
As mentioned earlier, it is conjectured that the bound above can be replaced by $2\Delta$ (see~\cite{AharoniBZ07}).

The following lemma determines the strong chromatic numbers of types $1$ and $2$ of the complete bipartite graph $K_{m,m}$ with $m$ vertices in each part. It illustrates that the first two levels of the hierarchy may differ.

\begin{lemma}\label{lemma:K_m,m}
For every positive integer $m$, it holds that $\chi_s^1(K_{m,m}) = \lceil 3m/2 \rceil$ and $\chi_s^2(K_{m,m}) = 2m$.
\end{lemma}

\begin{proof}
For a positive integer $m$, consider the complete bipartite graph $K_{m,m}$, and let $V_1$ and $V_2$ denote the parts of its bipartition. The case $m=1$ is immediate, so assume that $m \geq 2$. We first prove that $\chi_s^1(K_{m,m}) = \lceil 3m/2 \rceil$.

For the upper bound, we show that $K_{m,m}$ is strongly $(\ell,1)$-colorable for $\ell = \lceil 3m/2 \rceil$. Consider an arbitrary partition of the vertex set of $K_{m,m}$ into parts of size at most $\ell$, with at most one part of size smaller than $\ell$. For $m \geq 2$, it holds that $m < \ell < 2m$, hence such a partition consists of exactly two parts, which we denote by $Q_1$ and $Q_2$, where $|Q_1| = \ell$ and $|Q_2| = 2m-\ell = \lfloor m/2 \rfloor$. Define a coloring of $K_{m,m}$ as follows. First, assign $\ell$ distinct colors to the vertices of $Q_1$. Then, assign to the vertices of $Q_2 \cap V_1$ distinct colors that appear on the vertices of $Q_1 \cap V_1$, and assign to the vertices of $Q_2 \cap V_2$ distinct colors that appear on the vertices of $Q_1 \cap V_2$. This is possible because $|Q_2| = \lfloor m/2 \rfloor$, whereas $Q_1$ intersects each of $V_1$ and $V_2$ in at least $\ell - m = \lceil m/2 \rceil$ vertices. We obtain a proper $\ell$-coloring of $K_{m,m}$ that assigns distinct colors to the vertices of each of $Q_1$ and $Q_2$, as required.

For the lower bound, we show that $K_{m,m}$ is not strongly $(\ell,1)$-colorable for $\ell = \lceil 3m/2 \rceil - 1$. Note that $\ell \geq m$, and consider a partition of the vertex set of $K_{m,m}$ into a part $Q_1$ of size $\ell$ that contains $V_1$, and a part $Q_2$ that consists of the remaining vertices. Notice that the size of the latter satisfies $|Q_2| = 2m-\ell = \lfloor m/2 \rfloor+1 \leq \ell$. Any proper coloring of $K_{m,m}$ that assigns distinct colors to the vertices of each of $Q_1$ and $Q_2$ must also assign distinct colors to the vertices of $V_1 \cup Q_2$, since $V_1 \subseteq Q_1$ and every vertex of $Q_2$ is adjacent to every vertex of $V_1$. Therefore, the number of colors used by such a coloring is at least $|V_1|+|Q_2| = m+\lfloor m/2 \rfloor+1=\lfloor 3m/2 \rfloor +1>\ell$. This implies that there is no proper $\ell$-coloring of $K_{m,m}$ that assigns distinct colors to the vertices of each of $Q_1$ and $Q_2$, hence $K_{m,m}$ is not strongly $(\ell,1)$-colorable. By Lemma~\ref{lemma:monotone_chi_r}, we deduce that $\chi_s^1(K_{m,m}) \geq \ell+1 = \lceil 3m/2 \rceil$.

Finally, we observe that $\chi_s^2(K_{m,m}) = 2m$. The upper bound is trivial. For the lower bound, set $\ell=2m-1$, and consider the partition of the vertex set of $K_{m,m}$ into the parts $V_1$ and $V_2$, each of size $m \leq \ell$. Every proper coloring that assigns distinct colors to the vertices of each part must assign a different color to every vertex of $K_{m,m}$ and thus use more than $\ell$ colors. This shows that $K_{m,m}$ is not strongly $(\ell,2)$-colorable, which by Lemma~\ref{lemma:monotone_chi_r} implies that $\chi_s^2(K_{m,m}) \geq \ell+1 = 2m$.
\end{proof}

We end this section with the strong chromatic numbers of cycles and paths.
For a positive integer $m$, let $C_m$ and $P_m$ denote the cycle and path graphs on $m$ vertices, respectively. 
The following statement summarizes their strong chromatic numbers.

\begin{theorem}[\cite{FleischnerS92,Sachs93,Ohman}]\label{thm:chi_s_cycle_path}
For every positive integer $m$, it holds that
\[ \chi_s(C_{3m}) = \chi_s(C_{3m+2}) = 3 \mbox{~~~and~~~} \chi_s(C_{3m+1}) = 4,\]
and for every integer $m \geq 3$, it holds that $\chi_s(P_m) = 3$.
\end{theorem}

We next present the analogous statements for the strong chromatic numbers of types $1$ and $2$.

\begin{lemma}\label{lemma:chi_s^1_cycles}
For every integer $m \geq 3$, it holds that $\chi_s^1(C_m) = \chi_s^1(P_m) = 3$.
\end{lemma}

\begin{proof}
Since cycle and path graphs on at least three vertices have maximum degree $2$, Lemma~\ref{lemma:chi_s_1-D} implies that for every integer $m \geq 3$, it holds that $\chi_s^1(C_m) \geq 3$ and $\chi_s^1(P_m) \geq 3$. Combining Lemma~\ref{lemma:compare_chi_s} with Theorem~\ref{thm:chi_s_cycle_path} gives the matching upper bounds, except for cycles whose number of vertices is congruent to $1$ modulo $3$. To complete this remaining case, let $m$ be a positive integer, and let us show that $C_{3m+1}$ is strongly $(3,1)$-colorable. Consider a partition of the vertex set of $C_{3m+1}$ into $m$ parts of size exactly three and a singleton. By removing the vertex of the singleton from the cycle, we obtain a path on $3m$ vertices. By Theorem~\ref{thm:chi_s_cycle_path}, it holds that $\chi_s(P_{3m})=3$, hence the path admits a proper $3$-coloring that assigns distinct colors to the vertices of each part of size three in the given partition. This coloring can be extended to a proper $3$-coloring of $C_{3m+1}$ by assigning to the remaining vertex a color that differs from the colors of its two neighbors. Each color class of the extended coloring intersects each part of the given partition in at most one vertex, so we are done.
\end{proof}

\begin{lemma}\label{lemma:chi_s^2_cycles}
For every positive integer $m$, it holds that 
\[\chi_s^2(C_{3m}) = \chi_s^2(C_{3m+2}) = 3 \mbox{~~~and~~~} \chi_s^2(C_{3m+1}) = 4,\]
and for every integer $m \geq 3$, it holds that $\chi_s^2(P_m) = 3$.
\end{lemma}

\begin{proof}
By Lemma~\ref{lemma:compare_chi_s} and Theorem~\ref{thm:chi_s_cycle_path}, all the quantities in the lemma are upper bounded by the asserted values. The matching lower bounds follow from Lemma~\ref{lemma:compare_chi_s} together with Lemma~\ref{lemma:chi_s^1_cycles}, except for cycles whose number of vertices is congruent to $1$ modulo $3$. To complete this remaining case, let $m$ be a positive integer, and let us show that $C_{3m+1}$ is not strongly $(3,2)$-colorable. Indeed, by Theorem~\ref{thm:chi_s_cycle_path}, it holds that $\chi_s(C_{3m+1})=4$, implying that $C_{3m+1}$ is not strongly $3$-colorable. Therefore, after adding to this graph two isolated vertices, there exists a partition of the resulting vertex set into parts of size exactly three, such that no proper $3$-coloring assigns distinct colors to the vertices of each part. By removing the two added isolated vertices from the partition, we obtain a partition of the vertex set of $C_{3m+1}$ into parts of size at most three, with at most two parts of size strictly smaller than three, such that no proper $3$-coloring assigns distinct colors to the vertices of each part. This completes the proof.
\end{proof}

\section{Stochastic-Dominance Envy-Freeness under Conflict Constraints}\label{sec:SD}

In this section, we study feasible SD-EF1 allocations under conflict constraints. Our objective is to determine, for a given graph $G$, the positive integers $\ell$ for which $G$ admits a feasible SD-EF1 allocation to $\ell$ agents with respect to any weak order shared by all agents. We provide a precise characterization of these integers in terms of the strong chromatic number of type $1$ of $G$, as stated in Theorem~\ref{thm:intro_SD}. We also observe that this characterization fails to hold when the agents have heterogeneous weak orders.

\subsection{Reformulations of SD-EF1}

We begin with a couple of lemmas, providing alternative formulations of SD-EF1 allocations to agents with a common weak order.

\begin{lemma}\label{lemma:SD<=1}
Let $M$ be a finite set of items, let $\succeq$ be a weak order on $M$, let $\ell$ be a positive integer, and let $\calA = (A_1, \ldots, A_\ell)$ be an allocation of $M$ to $\ell$ agents.
The allocation $\calA$ is SD-EF1 with respect to the common weak order $\succeq$ if and only if for all $i,i' \in [\ell]$ and $\alpha \in M$, it holds that
\begin{eqnarray}\label{eq:SD<=1}
\Big | | A_i \cap S_\alpha | - |A_{i'} \cap S_\alpha | \Big | \leq 1,
\end{eqnarray}
where $S_\alpha = \{ x \in M \mid x \succeq \alpha\}$.
\end{lemma}

\begin{proof}
Suppose first that $\calA$ is SD-EF1 with respect to the common weak order $\succeq$, and let $i,i' \in [\ell]$ and $\alpha \in M$. By the SD-EF1 property, there exists a set $S \subseteq A_{i'}$ with $|S| \leq 1$, such that $A_i \succeq^{\SD} A_{i'} \setminus S$, implying that
\[|A_i \cap S_\alpha | \geq |(A_{i'} \setminus S) \cap S_\alpha| \geq |A_{i'} \cap S_\alpha |-1.\]
By symmetry, it further follows that $|A_{i'} \cap S_\alpha | \geq |A_{i} \cap S_\alpha |-1$. Combining these inequalities, we obtain the inequality in~\eqref{eq:SD<=1}, as required.

Conversely, suppose that $\calA$ satisfies the inequalities in~\eqref{eq:SD<=1}. To show that $\calA$ is SD-EF1 with respect to the common weak order $\succeq$, let $i,i' \in [\ell]$. If $A_{i'} = \emptyset$, then it clearly holds that $A_i \succeq^{\SD} A_{i'}$. Otherwise, let $\beta$ denote a most-preferred item in $A_{i'}$ with respect to $\succeq$, that is, an item $\beta \in A_{i'}$ satisfying $\beta \succeq x$ for all $x \in A_{i'}$. Fix $\alpha \in M$. If $\beta \in S_\alpha$, then
\[ |A_i \cap S_\alpha| \geq |A_{i'} \cap S_\alpha|-1 = |(A_{i'} \setminus \{\beta\}) \cap S_\alpha|,\]
where the inequality follows from~\eqref{eq:SD<=1}.
If $\beta \notin S_\alpha$, then since $\beta$ is most preferred in $A_{i'}$, no item of $A_{i'}$ belongs to $S_\alpha$, hence
\[ |A_i \cap S_\alpha| \geq 0 = |(A_{i'} \setminus \{\beta\}) \cap S_\alpha|.\]
This shows that $A_i \succeq^{\SD} A_{i'} \setminus \{\beta\}$, and we are done.
\end{proof}

The next reformulation of SD-EF1 allocations to agents with a common weak order resembles that of~\cite[Lemma~1]{BarmanELS25} and uses the following definition.

\begin{definition}\label{def:non-inc-order-part}
Let $M$ be a finite set, and let $t$ be a positive integer.
An {\em ordered partition} of $M$ into $t$ parts is a sequence $(Q_1, \ldots, Q_t)$ of subsets of $M$ that form a partition of $M$. It is called {\em non-increasing with respect to a weak order $\succeq$} if for all $j,j' \in [t]$ with $j < j'$ and for all $x \in Q_j$ and $y \in Q_{j'}$, it holds that $x \succeq y$.
\end{definition}

\begin{lemma}\label{lemma:SD_Pj}
Let $M$ be a finite set of items, let $\succeq$ be a weak order on $M$, let $\ell$ be a positive integer, and let $\calA = (A_1, \ldots, A_\ell)$ be an allocation of $M$ to $\ell$ agents. The allocation $\calA$ is SD-EF1 with respect to the common weak order $\succeq$ if and only if there exists an ordered partition $(Q_1, \ldots, Q_t)$ of $M$, non-increasing with respect to $\succeq$, into $t = \lceil |M|/\ell \rceil$ parts of size exactly $\ell$ except possibly $Q_t$, such that $|A_i \cap Q_j| = 1$ for all $i \in [\ell]$ and $j \in [t-1]$, and $|A_i \cap Q_t| \leq 1$ for all $i \in [\ell]$.
\end{lemma}

\begin{proof}
Suppose first that $\calA$ is SD-EF1 with respect to $\succeq$. By Lemma~\ref{lemma:SD<=1}, it follows that for all $i,i' \in [\ell]$ and $\alpha \in M$, it holds that $| | A_i \cap S_\alpha | - |A_{i'} \cap S_\alpha | | \leq 1$, where $S_\alpha = \{x \in M \mid x \succeq \alpha\}$.
In particular, by choosing $\alpha$ as a least-preferred item in $M$ with respect to $\succeq$, we derive that the sizes of the sets in $\calA$ differ by at most $1$ from each other. Since they form a partition of $M$, it follows that $|A_i| \in \{t-1,t\}$ for all $i \in [\ell]$.
For each $i \in [\ell]$, choose a permutation of the elements of $A_i$ that is non-increasing with respect to $\succeq$. For each $j \in [t]$, let $Q_j$ denote the set of all $j$th items in these permutations, over all agents that have $j$th item. We claim that the ordered partition $(Q_1, \ldots, Q_t)$ of $M$ satisfies the assertion of the lemma. First, using the sizes of the sets of $\calA$, it follows that $|Q_j| = \ell$ for all $j \in [t-1]$ and that $|Q_t| \leq \ell$. Furthermore, the construction guarantees that $|A_i \cap Q_j| = 1$ for all $i \in [\ell]$ and $j \in [t-1]$, and that $|A_i \cap Q_t| \leq 1$ for all $i \in [\ell]$. To verify that $(Q_1, \ldots, Q_t)$ is non-increasing with respect to $\succeq$, consider indices $j,j' \in [t]$ with $j < j'$ and items $x \in Q_j$ and $y \in Q_{j'}$. Let $i,i' \in [\ell]$ be the indices satisfying $x \in A_i$ and $y \in A_{i'}$. Suppose for contradiction that $x \succeq y$ does not hold, and thus $x \notin S_y$. Since $x$ is the $j$th item in the non-increasing permutation of $A_i$, it follows that $|A_i \cap S_y| \leq j-1$. On the other hand, since $y$ is the $j'$th item in the non-increasing permutation of $A_{i'}$, it follows that $|A_{i'} \cap S_y| \geq j' \geq j+1$. This implies that $| |A_i \cap S_y| - |A_{i'} \cap S_y| | \geq 2$, contradicting Lemma~\ref{lemma:SD<=1}.

For the other direction, suppose that there exists an ordered partition $(Q_1, \ldots, Q_t)$ of $M$ into $t$ parts as in the lemma. To show that $\calA$ is SD-EF1 with respect to $\succeq$, fix $i,i' \in [\ell]$ and $\alpha \in M$, and consider the set $S_\alpha$. Since the given ordered partition is non-increasing, there exist an index $t' \in [t]$ and a set $Q'_{t'} \subseteq Q_{t'}$ such that $S_\alpha = Q_1 \cup \cdots \cup Q_{t'-1} \cup Q'_{t'}$. The properties of the ordered partition imply that $|A_{i} \cap Q_j| = |A_{i'} \cap Q_j| = 1$ for all $j \in [t'-1]$, and that $|A_{i} \cap Q'_{t'}| \leq 1$ and $|A_{i'} \cap Q'_{t'}| \leq 1$. This implies that each of $A_i$ and $A_{i'}$ includes either $t'-1$ or $t'$ elements of $S_\alpha$, hence $| | A_i \cap S_\alpha | - |A_{i'} \cap S_\alpha | | \leq 1$. By Lemma~\ref{lemma:SD<=1}, it follows that $\calA$ is SD-EF1 with respect to the common weak order $\succeq$, as desired.
\end{proof}

\subsection{SD-EF1 under Conflict Constraints for Common Weak Orders}

Equipped with Lemma~\ref{lemma:SD_Pj}, we relate the existence of feasible SD-EF1 allocations for common weak orders to the strong chromatic number of type $1$. The following two statements establish the sufficient and necessary conditions asserted in Theorem~\ref{thm:intro_SD}.

\begin{theorem}\label{thm:SD_upper}
Let $G=(V,E)$ be a graph, and let $\ell$ be an integer satisfying $\ell \geq \chi_s^1(G)$.
Then, for every weak order $\succeq$ on $V$, the graph $G$ admits a feasible SD-EF1 allocation to $\ell$ agents with respect to the common weak order $\succeq$.
\end{theorem}

\begin{proof}
Let $G=(V,E)$ be a graph, let $\ell \geq \chi_s^1(G)$, and let $\succeq$ be a weak order on $V$. Choose a permutation of the vertices of $V$ that is non-increasing with respect to $\succeq$. Let $t = \lceil |V|/\ell \rceil$, and consider the ordered partition $(Q_1, \ldots, Q_t)$ of $V$ into $t$ parts of consecutive vertices along the chosen permutation, each of size exactly $\ell$ except possibly $Q_t$. By construction, this ordered partition is non-increasing with respect to $\succeq$. By Lemma~\ref{lemma:monotone_chi_r}, the assumption $\ell \geq \chi_s^1(G)$ implies that the graph $G$ is strongly $(\ell,1)$-colorable. Therefore, $G$ admits a proper $\ell$-coloring such that each color class intersects each of $Q_1, \ldots, Q_t$ in at most one vertex. Since $|Q_j|=\ell$ for all $j \in [t-1]$, each color class intersects each of the parts $Q_1, \ldots, Q_{t-1}$ in exactly one vertex. Let $\calA = (A_1, \ldots, A_\ell)$ be a sequence of the color classes of this coloring, ordered arbitrarily. It holds that $|A_i \cap Q_j|=1$ for all $i \in [\ell]$ and $j \in [t-1]$, and $|A_i \cap Q_t| \leq 1$ for all $i \in [\ell]$. By Lemma~\ref{lemma:SD_Pj}, $\calA$ is an SD-EF1 allocation of $V$ with respect to the common weak order $\succeq$. Since the sets $A_1, \ldots, A_\ell$ are independent in $G$, this allocation is feasible, and the proof is complete.
\end{proof}

\begin{theorem}\label{thm:SD_lower}
Let $G=(V,E)$ be a graph, and let $\ell$ be a positive integer satisfying $\ell < \chi_s^1(G)$.
Then, there exists a weak order $\succeq$ on $V$, such that the graph $G$ admits no feasible SD-EF1 allocation to $\ell$ agents with respect to the common weak order $\succeq$.
\end{theorem}

\begin{proof}
Let $G=(V,E)$ be a graph, and let $\ell$ be a positive integer satisfying $\ell < \chi_s^1(G)$. Since $G$ is not strongly $(\ell,1)$-colorable, there exists a partition of $V$ into $t = \lceil |V| / \ell \rceil$ parts $Q_1, \ldots, Q_t$, with $|Q_j|=\ell$ for all $j \in [t-1]$ and $|Q_t| \leq \ell$, such that no proper $\ell$-coloring of $G$ assigns distinct colors to the vertices of each part. Define a weak order $\succeq$ on $V$, such that for all $x,y \in V$, it holds that $x \succeq y$ if and only if the indices $j,j' \in [t]$ for which $x \in Q_j$ and $y \in Q_{j'}$ satisfy $j \leq j'$. Thus, all vertices in the same part $Q_j$ are tied under $\succeq$, while vertices in parts with smaller indices are strictly preferred under $\succeq$ to vertices in parts with larger indices.

We claim that $G$ admits no feasible SD-EF1 allocation to $\ell$ agents with respect to the common weak order $\succeq$. Suppose otherwise, and let $\calA = (A_1, \ldots, A_\ell)$ be such an allocation. Since $\calA$ is feasible, its sets are independent in $G$, and thus induce a proper $\ell$-coloring of $G$. Since $\calA$ is SD-EF1 with respect to $\succeq$, Lemma~\ref{lemma:SD_Pj} implies that there exists a non-increasing ordered partition $(Q'_1, \ldots, Q'_t)$ of $V$ into $t$ parts of size exactly $\ell$ except possibly $Q'_t$, such that $|A_i \cap Q'_j| = 1$ for all $i \in [\ell]$ and $j \in [t-1]$, and $|A_i \cap Q'_t| \leq 1$ for all $i \in [\ell]$. By the definition of $\succeq$, its equivalence classes, ordered from most preferred to least preferred, are precisely $Q_1, \ldots, Q_t$. Since $|Q_j|=|Q'_j|=\ell$ for all $j \in [t-1]$, the non-increasing property of $(Q'_1, \ldots, Q'_t)$ forces, successively, that $Q'_j = Q_j$ for all $j \in [t-1]$, and thus $Q'_t = Q_t$ as well. Consequently, $|A_i \cap Q_j| \leq 1$ for all $i \in [\ell]$ and $j \in [t]$. Therefore, the proper $\ell$-coloring of $G$ induced by $\calA$ assigns distinct colors to the vertices of each of $Q_1, \ldots, Q_t$, contradicting the choice of this partition of $V$. This completes the proof.
\end{proof}

\subsection{SD-EF1 under Conflict Constraints for Heterogeneous Weak Orders}

Theorem~\ref{thm:SD_upper} shows that a graph admits a feasible SD-EF1 allocation to agents with any common weak order as long as the number of agents is at least its strong chromatic number of type $1$. The assumption that the weak order is common to all agents is essential in this result. This follows from related constructions presented in~\cite[Theorem~6]{CooksonE025} and in~\cite[Theorem~1]{BarmanELS25}. In the former construction, the conflict graph is a perfect matching on eight vertices, and in the latter, it is the same graph with one vertex removed. For completeness, we present here the smaller construction in full detail.

\begin{proposition}\label{prop:3K2+K1}
There exist a graph $G=(V,E)$ with $\chi_s^1(G)=2$ and two weak orders $\succeq_1$ and $\succeq_2$ on $V$, such that $G$ admits no feasible SD-EF1 allocation to two agents associated with $\succeq_1$ and $\succeq_2$.
\end{proposition}

\begin{proof}
Consider the graph $G$ defined as a perfect matching on six vertices together with an isolated vertex, namely, the vertex set is $V = \{v_1, \ldots, v_7\}$, and the edges are $\{v_{2i-1},v_{2i}\}$ for $i \in [3]$. We first observe that $\chi_s^1(G)=2$. Indeed, we have $2 \leq \chi_s^1(G) \leq \chi_s^\infty(G) \leq 2$, where the three inequalities follow, respectively, from Lemmas~\ref{lemma:chi_s_1-D},~\ref{lemma:compare_chi_s}, and~\ref{lemma:chi_s^inf<=C}. We let $\succeq_1$ and $\succeq_2$ be the two strict orders on $V$ defined as follows.
\[\succeq_1:~~~ v_1 \succ v_3 \succ v_2 \succ v_5 \succ v_4 \succ v_7 \succ v_6.\]
\[\succeq_2:~~~ v_1 \succ v_5 \succ v_2 \succ v_3 \succ v_6 \succ v_7 \succ v_4.\]

Any allocation of $G$ to two agents can be represented by a vector $x \in \{0,1\}^7$, where for each $i \in [7]$, $x_i=1$ if $v_i$ is assigned to agent $1$, and $x_i=0$ if $v_i$ is assigned to agent $2$. Further, a feasible allocation to two agents must assign one vertex from each edge to each agent, implying that
\begin{eqnarray}\label{eq:1}
x_1+x_2 = x_3+x_4 = x_5+x_6=1.
\end{eqnarray}

For the SD-EF1 property, we require that each agent receives, for every $j \in [3]$, at least $j$ of the top $2j$ items according to her order. Indeed, if she receives at most $j-1$ of these items, then the other agent receives at least $j+1$ of them, which remains at least $j$ even after the removal of one item. In particular, agent $1$ receives at least one of $\{v_1,v_3\}$, and agent $2$ receives at least one of $\{v_1,v_5\}$.
It thus follows that
\begin{eqnarray}\label{eq:2}
x_1+x_3 \geq 1 \mbox{ ~~~and~~~ } x_1+x_5 \leq 1.
\end{eqnarray}
Moreover, the two agents have the same set $\{v_1,v_2,v_3,v_5\}$ of top four items. Since each agent must receive at least two items from this set, it follows that $x_1+x_2+x_3+x_5=2$. Combining this with the equality $x_1+x_2=1$, given in~\eqref{eq:1}, yields that
\begin{eqnarray}\label{eq:3}
x_3+x_5=1.
\end{eqnarray}
Finally, considering the top six items of each agent, we obtain that $x_1+x_3+x_2+x_5+x_4+x_7 \geq 3$ for agent $1$, and that $x_1+x_5+x_2+x_3+x_6+x_7 \leq 3$ for agent $2$. Using $x_1+x_2+x_3+x_5=2$, it follows that
\begin{eqnarray}\label{eq:4}
x_4+x_7 \geq 1 \mbox{ ~~~and~~~ } x_6+x_7 \leq 1.
\end{eqnarray}

We claim that the above conditions are inconsistent. Indeed, if $x_1 = 1$, then it follows from~\eqref{eq:1} that $x_2=0$, and from~\eqref{eq:2} that $x_5=0$. It then follows from~\eqref{eq:3} that $x_3= 1$, and by applying~\eqref{eq:1} again, we obtain that $x_4=0$ and $x_6=1$. Similarly, if $x_1 = 0$, then it follows from~\eqref{eq:1} that $x_2=1$, and from~\eqref{eq:2} that $x_3=1$. It then follows from~\eqref{eq:3} that $x_5= 0$, and from~\eqref{eq:1} we obtain again that $x_4=0$ and $x_6=1$. However, these values of $x_4$ and $x_6$ leave no choice for $x_7$ satisfying~\eqref{eq:4}, so the described instance admits no feasible SD-EF1 allocation, completing the proof.
\end{proof}

\section{Envy-Freeness under Conflict Constraints for Additive Valuations}\label{sec:EF}

In this section, we study feasible EF1 and EF[1,1] allocations under conflict constraints for agents sharing common preferences. We obtain sufficient conditions on the number of agents for the existence of such allocations in terms of the strong chromatic numbers of types $1$ and $2$, respectively. We also show that these conditions are not necessary in general.

\subsection{EF1 and EF[1,1] from Ordered Partitions}

We begin with a couple of lemmas that offer sufficient conditions for EF1 and EF[1,1] fairness in the common valuation setting. To state them, we extend the notion of non-increasing ordered partitions, given in Definition~\ref{def:non-inc-order-part} for weak orders, to the context of additive valuations. An ordered partition of a set $M$ is called {\em non-increasing with respect to an additive valuation $w:\calP(M) \to \R$} if it is non-increasing with respect to the weak order on $M$ induced by $w$.

\begin{lemma}\label{lemma:EF1_Pj}
Let $M$ be a finite set of items, and let $w:\calP(M) \to \R$ be a non-decreasing additive valuation. 
For a positive integer $\ell$, let $(Q_1, \ldots, Q_t)$ be an ordered partition of $M$, non-increasing with respect to the valuation $w$, into $t = \lceil |M|/\ell \rceil$ parts of size exactly $\ell$ except possibly $Q_t$. Then every allocation $\calA = (A_1, \ldots, A_\ell)$ of $M$ to $\ell$ agents that satisfies $|A_i \cap Q_j| = 1$ for all $i \in [\ell]$ and $j \in [t-1]$, and $|A_i \cap Q_t| \leq 1$ for all $i \in [\ell]$, is EF1 with respect to the common valuation $w$.
\end{lemma}

\begin{proof}
Let $\succeq$ denote the weak order on $M$ induced by the non-decreasing additive valuation $w$.
Since the ordered partition $(Q_1, \ldots, Q_t)$ of $M$ is non-increasing with respect to $w$, it is also non-increasing with respect to $\succeq$.
Let $\calA = (A_1, \ldots, A_\ell)$ be an allocation of $M$ to $\ell$ agents that satisfies $|A_i \cap Q_j| = 1$ for all $i \in [\ell]$ and $j \in [t-1]$, and $|A_i \cap Q_t| \leq 1$ for all $i \in [\ell]$. By Lemma~\ref{lemma:SD_Pj}, $\calA$ is SD-EF1 with respect to the common weak order $\succeq$. Since $w$ is non-decreasing, we can apply Lemma~\ref{lemma:SD->EF1} and derive that it is also EF1 with respect to the common valuation $w$, as required.
\end{proof}

\begin{lemma}\label{lemma:EF11_P_Q}
Let $M$ be a finite set of items, let $w:\calP(M) \to \R$ be an additive valuation, and consider a partition of $M$ into two sets $M^+$ and $M^-$, consisting of goods and chores, respectively, under $w$. For a positive integer $\ell$, let $(Q^+_1, \ldots, Q^+_t)$ be an ordered partition of $M^+$, non-increasing with respect to the valuation $w$ restricted to $M^+$, into $t = \lceil |M^+|/\ell \rceil$ parts of size exactly $\ell$ except possibly $Q^+_t$, and let $(Q^-_1, \ldots, Q^-_{t'})$ be an ordered partition of $M^-$, non-increasing with respect to the negated valuation $-w$ restricted to $M^-$, into $t' = \lceil |M^-|/\ell \rceil$ parts of size exactly $\ell$ except possibly $Q^-_{t'}$. Then every allocation $\calA = (A_1, \ldots, A_\ell)$ of $M$ to $\ell$ agents that satisfies
\begin{eqnarray}\label{eq:EF11_P_Q}
|A_i \cap Q^+_j| \leq 1 \mbox{~~~and~~~} |A_i \cap Q^-_{j'}| \leq 1 \mbox{~~~for all~~~} i \in [\ell],~j \in [t], \mbox{~and~} j' \in [t']
\end{eqnarray}
is EF[1,1] with respect to the common valuation $w$.
\end{lemma}

\begin{proof}
Let $\calA = (A_1, \ldots, A_\ell)$ be an allocation of $M$ to $\ell$ agents satisfying~\eqref{eq:EF11_P_Q}. Since $\calA$ is an allocation of $M$, and since all parts except possibly $Q^+_t$ and $Q^-_{t'}$ have size exactly $\ell$, the inequalities in~\eqref{eq:EF11_P_Q} are in fact equalities for all $i \in [\ell]$, $j \in [t-1]$, and $j' \in [t'-1]$. Let $\calB = (B_1, \ldots, B_\ell)$ and $\calC = (C_1, \ldots, C_\ell)$ be the allocations to $\ell$ agents of $M^+$ and $M^-$, respectively, defined by $B_i = A_i \cap M^+$ and $C_i = A_i \cap M^-$ for all $i \in [\ell]$. By Lemma~\ref{lemma:EF1_Pj}, applied to the non-decreasing restriction of $w$ to $M^+$ with the ordered partition $(Q^+_1, \ldots, Q^+_t)$, the allocation $\calB$ is EF1 with respect to $w$ restricted to $M^+$. Similarly, by Lemma~\ref{lemma:EF1_Pj}, applied to the non-decreasing restriction of $-w$ to $M^-$ with the ordered partition $(Q^-_1, \ldots, Q^-_{t'})$, the allocation $\calC$ is EF1 with respect to $-w$ restricted to $M^-$, and hence also with respect to $w$ restricted to $M^-$. Therefore, by Lemma~\ref{lemma:EF1->EF11}, the allocation $\calA$ is EF[1,1] with respect to the common valuation $w$, which completes the proof.
\end{proof}

\subsection{EF1 and EF[1,1] under Conflict Constraints for Common Valuations}

We now prove Theorem~\ref{thm:intro_EF1_upper}, which relates the existence of feasible EF1 allocations for common monotone additive valuations to the strong chromatic number of type $1$.

\begin{proof}[ of Theorem~\ref{thm:intro_EF1_upper}]
Let $G=(V,E)$ be a graph, let $\ell \geq \chi_s^1(G)$, and let $w:\calP(V) \to \R$ be a monotone additive valuation.
We may assume that $w$ is non-decreasing, since the non-increasing case follows by replacing $w$ with the negated valuation $-w$. Let $\succeq$ denote the weak order on $V$ induced by $w$. By Theorem~\ref{thm:SD_upper}, using $\ell \geq \chi_s^1(G)$, the graph $G$ admits a feasible SD-EF1 allocation to $\ell$ agents with respect to the common weak order $\succeq$. Since $w$ is non-decreasing, Lemma~\ref{lemma:SD->EF1} implies that this allocation is also EF1 with respect to the common valuation $w$. Thus, $G$ admits a feasible EF1 allocation to $\ell$ agents with respect to the common valuation $w$, as required.
\end{proof}

We next prove Theorem~\ref{thm:intro_EF11_upper}, which relates the existence of feasible EF[1,1] allocations for common arbitrary additive valuations to the strong chromatic number of type $2$.

\begin{proof}[ of Theorem~\ref{thm:intro_EF11_upper}]
Let $G=(V,E)$ be a graph, let $\ell \geq \chi_s^2(G)$, and let $w:\calP(V) \to \R$ be an additive valuation. We may assume that $w$ is not monotone. Indeed, if $w$ is monotone, we can apply Theorem~\ref{thm:intro_EF1_upper}, using $\ell \geq \chi_s^2(G) \geq \chi_s^1(G)$, to obtain a feasible EF1 allocation of $G$ to $\ell$ agents with respect to the common valuation $w$. This allocation is also EF[1,1].

Consider a partition of $V$ into two sets $V^+$ and $V^-$, consisting of goods and chores, respectively, under $w$, where vertices of value $0$ are placed arbitrarily. Since $w$ is not monotone, both $V^+$ and $V^-$ are non-empty. Choose a permutation of the vertices of $V^+$ that is non-increasing with respect to $w$. Let $t = \lceil |V^+|/\ell \rceil$, and consider the ordered partition $(Q^+_1, \ldots, Q^+_t)$ of $V^+$ into $t$ parts of consecutive vertices along the chosen permutation, each of size exactly $\ell$ except possibly $Q^+_t$. By construction, this ordered partition is non-increasing with respect to the restriction of $w$ to $V^+$. Similarly, choose a permutation of the vertices of $V^-$ that is non-increasing with respect to $-w$. Let $t' = \lceil |V^-|/\ell \rceil$, and consider the ordered partition $(Q^-_1, \ldots, Q^-_{t'})$ of $V^-$ into $t'$ parts of consecutive vertices along the chosen permutation, each of size exactly $\ell$ except possibly $Q^-_{t'}$. As before, this ordered partition is non-increasing with respect to the restriction of $-w$ to $V^-$.

The assumption $\ell \geq \chi_s^2(G)$ implies, by Lemma~\ref{lemma:monotone_chi_r}, that the graph $G$ is strongly $(\ell,2)$-colorable. The sets $Q^+_1, \ldots, Q^+_t, Q^-_1, \ldots, Q^-_{t'}$ form a partition of $V$ into parts of size at most $\ell$, with at most two parts of size strictly smaller than $\ell$. Therefore, $G$ admits a proper $\ell$-coloring such that each color class intersects each of $Q^+_1, \ldots, Q^+_t, Q^-_1, \ldots, Q^-_{t'}$ in at most one vertex. Let $\calA = (A_1, \ldots, A_\ell)$ be a sequence of the color classes of this coloring, ordered arbitrarily. By Lemma~\ref{lemma:EF11_P_Q}, $\calA$ is an EF[1,1] allocation of $V$ with respect to the common valuation $w$. Since the sets $A_1, \ldots, A_\ell$ are independent in $G$, this allocation is feasible, so we are done.
\end{proof}

It is natural to ask whether the dependence on the strong chromatic number of type $2$ in Theorem~\ref{thm:intro_EF11_upper} can be weakened to the strong chromatic number of type $1$, as in Theorem~\ref{thm:intro_EF1_upper}. The following simple result gives a negative answer.

\begin{proposition}\label{prop:Km,m}
For every integer $m \geq 2$, there exist a graph $G=(V,E)$ satisfying $\chi_s^1(G)=\lceil 3m/2 \rceil$ and $\chi_s^2(G)= 2m$, and an additive valuation $w: \calP(V) \to \R$, such that $G$ admits no feasible EF[1,1] allocation to fewer than $2m$ agents with respect to the common valuation $w$.
\end{proposition}

\begin{proof}
For an integer $m \geq 2$, let $G=(V,E)$ be the complete bipartite graph $K_{m,m}$, and let $V_1$ and $V_2$ denote the parts of its bipartition. By Lemma~\ref{lemma:K_m,m}, we have $\chi_s^1(G)=\lceil 3m/2 \rceil$ and $\chi_s^2(G)= 2m$. Let $w: \calP(V) \to \R$ be the additive valuation defined by $w(v) = 1$ for all $v \in V_1$, and $w(v)=-1$ for all $v \in V_2$. We claim that $G$ admits no feasible EF[1,1] allocation to fewer than $2m$ agents with respect to the common valuation $w$. Indeed, any feasible allocation of $G$ assigns to each agent either a subset of $V_1$ or a subset of $V_2$. When the number of agents is smaller than $2m$, some agent receives at least two vertices. If the assignment of this agent is contained in $V_1$, then its value under $w$ is at least $2$, and it remains at least $1$ after removing at most one vertex. However, there exists an agent whose assignment is contained in $V_2$, and the value of such an assignment under $w$ is non-positive, even after removing at most one vertex. Hence the allocation is not EF[1,1] under the common valuation $w$. The case where an assignment with at least two vertices is contained in $V_2$ is handled similarly. This completes the proof.
\end{proof}

\subsection{Below the Strong Chromatic Numbers}

We now study the necessity of the sufficient conditions given in Theorems~\ref{thm:intro_EF1_upper} and~\ref{thm:intro_EF11_upper}. We show that they are necessary for graphs in which the relevant strong chromatic number is at most three, but are not necessary already for certain graphs with the corresponding strong chromatic number equal to four.

\subsubsection{Strong Chromatic Number at Most Three}

We claim that for every graph whose strong chromatic number of type $1$ or $2$ is at most three, there exists a common non-decreasing additive valuation for which no feasible EF1 allocation exists with fewer agents than the corresponding strong chromatic number. The cases where this number is one or two are immediate. Indeed, the former is trivial, and in the latter, the graph spans an edge, so it admits no feasible allocation to a single agent. The following result handles the remaining case.

\begin{proposition}\label{prop:chi_s^1=3}
Let $G=(V,E)$ be a graph satisfying $\chi_s^1(G)=3$ or $\chi_s^2(G)=3$ (or both). Then there exists a non-decreasing additive valuation $w: \calP(V) \to \R$, such that $G$ admits no feasible EF1 allocation to fewer than three agents with respect to the common valuation $w$.
\end{proposition}

\begin{proof}
Let $G=(V,E)$ be a graph satisfying $\chi_s^1(G)=3$ or $\chi_s^2(G)=3$. Since $G$ spans at least one edge, it clearly admits no feasible allocation to a single agent, so it is sufficient to consider the case of two agents. We first observe that the maximum degree of $G$ is at least $2$. Indeed, otherwise every connected component of $G$ has at most two vertices, hence by Lemmas~\ref{lemma:compare_chi_s} and~\ref{lemma:chi_s^inf<=C}, it follows that $\chi_s^1(G) \leq \chi_s^2(G) \leq \chi_s^\infty(G) \leq 2$, contradicting the assumption that $\chi_s^1(G)=3$ or $\chi_s^2(G)=3$. Now, let $x$ be a vertex of degree at least $2$ in $G$, let $y$ and $z$ be two of its neighbors, and let $w: \calP(V) \to \R$ be the non-decreasing additive valuation defined by $w(y)=w(z)=1$ and by $w(v)=0$ for every $v \in V \setminus \{y,z\}$. Any feasible allocation of $G$ to two agents must separate the vertex $x$ from its neighbors $y$ and $z$. In such a feasible allocation, the set that includes $x$ has value $0$, whereas the other set includes both $y$ and $z$ and has value $2$. Removing a vertex from the latter leaves its value at least $1$, so the allocation is not EF1 with respect to $w$. This shows that $G$ admits no feasible EF1 allocation to two agents with respect to the common valuation $w$, so we are done.
\end{proof}

\subsubsection{Strong Chromatic Number Four}

We turn to showing that the sufficient conditions in Theorems~\ref{thm:intro_EF1_upper} and~\ref{thm:intro_EF11_upper} for feasible EF1 and EF[1,1] allocations, respectively, are not necessary in general. We start with the EF1 fairness notion and prove Theorem~\ref{thm:intro_EF1_lower}, which asserts that there exists a graph $G$ with $\chi_s^1(G) = 4$ that admits a feasible EF1 allocation to three agents with respect to every common monotone additive valuation.

\begin{proof}[ of Theorem~\ref{thm:intro_EF1_lower}]
Let $G=(V,E)$ be the graph obtained from the cycle $C_4$ on the vertex set $\{v_1, v_2, v_3, v_4\}$ with the natural order along the cycle by adding an isolated vertex $z$. We first show that $\chi_s^1(G) = 4$. Since the maximum size of a connected component in $G$ is $4$, it follows from Lemmas~\ref{lemma:compare_chi_s} and~\ref{lemma:chi_s^inf<=C} that $\chi_s^1(G) \leq \chi_s^\infty(G) \leq 4$. On the other hand, consider the partition of $V$ into the parts $\{v_1,v_3,z\}$ and $\{v_2,v_4\}$, both of size at most three and only one of size strictly smaller than three. Observe that no proper $3$-coloring of $G$ assigns distinct colors to the vertices of each part. Indeed, a proper coloring that assigns distinct colors to the vertices of each part would force the four vertices of the cycle to receive pairwise distinct colors: this holds for the pairs $v_1,v_3$ and $v_2,v_4$ by the partition, and for all other pairs by the edges of $C_4$. This shows that $G$ is not strongly $(3,1)$-colorable, which by Lemma~\ref{lemma:monotone_chi_r} implies that $\chi_s^1(G) \geq 4$, and thus $\chi_s^1(G) = 4$.

Now, let $w:\calP(V) \to \R$ be a monotone additive valuation. We show that $G$ admits a feasible EF1 allocation to three agents with respect to the common valuation $w$. We may assume that $w$ is non-decreasing, because the non-increasing case follows by applying the result to $-w$. We may further assume, by symmetry, that $v_1$ has minimum value under $w$ among the vertices of the cycle, and that $w(v_2) \leq w(v_4)$.

Consider the sets  $A_1 = \{v_1,v_3\}$, $A_2 = \{v_2\}$, and $A_3 = \{v_4\}$. We define an allocation $\calA$ of $G$ to three agents as follows. If $w(A_1) \leq w(A_2)$, then let $\calA = (A_1 \cup \{z\},A_2,A_3)$, and otherwise, let $\calA = (A_1,A_2 \cup \{z\},A_3)$. The allocation $\calA$ is feasible, because $v_1$ and $v_3$ are not adjacent and $z$ is isolated in $G$. To show that it is EF1 with respect to $w$, let $T = \min(w(v_1)+w(v_3),w(v_2))$. Observe that each of the sets $A_1,A_2,A_3$ has value at least $T$ under $w$, and since $w$ is non-decreasing, this property is preserved after adding $z$ to one of $A_1$ and $A_2$. It is thus sufficient to show that one may remove a vertex from each allocated set to obtain a set with value at most $T$ under $w$. For $A_3$, remove $v_4$. For the allocated set that includes $z$, removing $z$ leaves a set among $A_1$ and $A_2$ of value exactly $T$. The remaining allocated set is either $A_1$ or $A_2$. If it is $A_2$, remove $v_2$, and if it is $A_1$, remove $v_3$ to get value $w(v_1)$, which is the minimum value of a vertex along the cycle and is thus at most $T$. This completes the proof.
\end{proof}

We proceed with the more nuanced case of EF[1,1] fairness. We show that there exists a graph $G$ with $\chi_s^2(G)=4$ that admits a feasible EF[1,1] allocation to three agents with respect to every common additive valuation. To this end, we consider the family of cycle graphs $C_{3m+1}$ for positive integers $m$, whose strong chromatic numbers of types $1$ and $2$ are three and four, respectively (see Lemmas~\ref{lemma:chi_s^1_cycles} and~\ref{lemma:chi_s^2_cycles}). Theorem~\ref{thm:intro_EF1_upper} implies that these graphs admit a feasible EF1 allocation to three agents with respect to every common monotone additive valuation. For $m=2$, we strengthen this conclusion by showing that $C_7$ admits a feasible EF[1,1] allocation to three agents with respect to every common additive valuation. Note that such a result does not hold for $m=1$, as follows from the proof of Proposition~\ref{prop:Km,m}.

We start with the following theorem, which handles the cycle graphs $C_{3m+1}$ where the number of goods under the common additive valuation is either $0$ or $1$ modulo $3$.

\begin{theorem}\label{thm:C_3m+1_t}
For a positive integer $m$, let $C_{3m+1} = (V,E)$, and let $w: \calP(V) \to \R$ be an additive valuation.
Suppose that $V$ can be written as a disjoint union $V = V^+ \cup V^-$, where $V^+$ and $V^-$ consist of goods and chores, respectively, under $w$, and $|V^+| \not\equiv 2 \pmod{3}$. Then $C_{3m+1}$ admits a feasible EF[1,1] allocation to three agents with respect to the common valuation $w$.
\end{theorem}

\begin{proof}
We may assume that both $V^+$ and $V^-$ are non-empty. Otherwise, the valuation $w$ is monotone, and by Lemma~\ref{lemma:chi_s^1_cycles} we have $\chi_s^1(C_{3m+1})=3$, hence Theorem~\ref{thm:intro_EF1_upper} gives a feasible EF1 allocation of $C_{3m+1}$ to three agents with respect to the common valuation $w$. Since $w$ is monotone, this allocation is also EF[1,1].

Pick permutations of the vertices of $V^+$ and $V^-$ that are non-increasing with respect to $w$ and $-w$, respectively. Consider the ordered partitions $(Q^+_1, \ldots, Q^+_t)$ of $V^+$ and $(Q^-_1, \ldots, Q^-_{t'})$ of $V^-$ into parts of consecutive vertices in these permutations, each of size exactly three except possibly the parts $Q^+_t$ and $Q^-_{t'}$ that may have smaller sizes. Since $|V| \equiv 1 \pmod{3}$, the assumption $|V^+| \not\equiv 2 \pmod{3}$ implies that one of $|V^+|$ and $|V^-|$ is congruent to $0$ modulo $3$, while the other is congruent to $1$ modulo $3$. It thus follows that all the parts $Q^+_1, \ldots, Q^+_t, Q^-_1, \ldots, Q^-_{t'}$ have size three except one, either $Q^+_t$ or $Q^-_{t'}$, which consists of a single vertex. By Lemma~\ref{lemma:chi_s^1_cycles}, it holds that $\chi_s^1(C_{3m+1})=3$, hence there exists a proper $3$-coloring of $C_{3m+1}$ such that each color class intersects each of $Q^+_1, \ldots, Q^+_t, Q^-_1, \ldots, Q^-_{t'}$ in at most one vertex. Let $\calA = (A_1,A_2,A_3)$ be an allocation whose sets are the color classes of this coloring, ordered arbitrarily. Note that $|A_i \cap Q^+_j| \leq 1$ and $|A_i \cap Q^-_{j'}| \leq 1$ for all $i \in [3]$, $j \in [t]$, and $j' \in [t']$. This allocation of $C_{3m+1}$ is clearly feasible, and by Lemma~\ref{lemma:EF11_P_Q}, it is EF[1,1] with respect to the common valuation $w$. This completes the proof.
\end{proof}

Now, for the cycle on seven vertices, we prove the following result, confirming Theorem~\ref{thm:intro_EF11_lower}.

\begin{theorem}\label{thm:C7}
Let $C_7=(V,E)$, and let $w: \calP(V) \to \R$ be an additive valuation.
Then $C_7$ admits a feasible EF[1,1] allocation to three agents with respect to the common valuation $w$.
\end{theorem}

\begin{proof}
Let $w: \calP(V) \to \R$ be an additive valuation. Write $V = V^+ \cup V^-$, where $V^+$ and $V^-$ are disjoint sets of goods and chores, respectively, under $w$. If $|V^+| \not\equiv 2 \pmod{3}$, then we can apply Theorem~\ref{thm:C_3m+1_t} to obtain that $C_7$ admits a feasible EF[1,1] allocation to three agents with respect to the common valuation $w$. It thus remains to handle the case where $|V^+|$ is either $2$ or $5$. Since EF[1,1] is invariant under replacing $w$ by $-w$, we may assume that $|V^+| = 5$. It follows that there exists a non-increasing ordered partition $(Q^+_1,Q^+_2)$ of $V^+$ with respect to the valuation $w$ restricted to $V^+$, satisfying $|Q^+_1|=3$ and $|Q^+_2|=2$, and that $|V^-|=2$.

Suppose first that there exists a proper $3$-coloring of $C_7$ that assigns distinct colors to the vertices of each of the parts $Q^+_1,Q^+_2,V^-$. By Lemma~\ref{lemma:EF11_P_Q}, applied with the ordered partitions $(Q^+_1,Q^+_2)$ and $(V^-)$ of goods and chores, respectively, a sequence of the color classes of this coloring forms a feasible EF[1,1] allocation of $C_7$ to three agents with respect to the common valuation $w$. We thus may assume that such a coloring does not exist. We observe that this implies that the sets $Q^+_1,Q^+_2,V^-$ are independent in $C_7$. Indeed, otherwise there exists an edge $e$ in $C_7$ whose vertices lie in the same set. By Lemma~\ref{lemma:chi_s^2_cycles}, the path $P_7$ is strongly $(3,2)$-colorable. It follows that the path obtained from $C_7$ by removing the edge $e$ admits a proper $3$-coloring that assigns distinct colors to the vertices of each of $Q^+_1,Q^+_2,V^-$, and since the endpoints of $e$ lie in the same set, the coloring is proper in $C_7$ as well. This contradicts our assumption that no such coloring exists.

Now, it is not difficult to observe that, up to an automorphism of the cycle, there are precisely two partitions of the vertex set of $C_7$ into three independent sets of sizes $3$, $2$, and $2$. Indeed, labeling the vertices by $v_1,v_2,\ldots,v_7$ along the cycle, the independent set of size $3$ can be mapped under such an automorphism to $\{v_1,v_3,v_5\}$. The remaining four vertices can be partitioned into two independent sets of size $2$ in exactly two ways, namely, as $\{v_2,v_6\}$ and $\{v_4,v_7\}$, or as $\{v_4,v_6\}$ and $\{v_2,v_7\}$. The former partition admits a proper $3$-coloring of $C_7$ that assigns distinct colors to the vertices of each part, as follows from the coloring with the color classes $\{v_1,v_4,v_6\}$, $\{v_3,v_7\}$, and $\{v_2,v_5\}$. We therefore may assume that $Q^+_1,Q^+_2,V^-$ are the sets of the latter partition. Moreover, by reversing the cycle if needed, we may also assume that
\[Q^+_1 = \{v_1,v_3,v_5\},~~Q^+_2 = \{v_4,v_6\},~~\mbox{and}~~V^- = \{v_2,v_7\}.\]

To show the existence of a feasible EF[1,1] allocation of $C_7$ to three agents with respect to the common valuation $w$, let $T = \max(w(v_4),w(v_6))$. Note that the value of every vertex in $Q^+_1$ is at least $T$, because the ordered partition $(Q^+_1,Q^+_2)$ of $V^+$ is non-increasing. We consider two cases.
In the first, suppose that
\[w(v_5) + \max(w(v_2),w(v_7)) \geq T.\]
Consider the allocation $(A_1,A_2,A_3)$ defined by
\[A_1 = \{v_3,v_6\},~~A_2=\{v_1,v_4\},~~\mbox{and}~~A_3=\{v_2,v_5,v_7\}.\]
These sets are clearly independent in $C_7$.
To verify that the allocation is EF[1,1] with respect to the common valuation $w$, it suffices to verify that, for each $i \in [3]$, one can remove at most one vertex from $A_i$ to obtain a set of value at most $T$, and one can remove at most one vertex from $A_i$ to obtain a set of value at least $T$. Then, for any two agents $i,i' \in [3]$, we can use the former removal from $A_{i'}$ and the latter removal from $A_{i}$ to derive the EF[1,1] condition.
Now, for $A_1$, remove $v_3$ to get value $w(v_6) \leq T$, and remove $v_6$ to get value $w(v_3) \geq T$; for $A_2$, remove $v_1$ to get value $w(v_4) \leq T$, and remove $v_4$ to get value $w(v_1) \geq T$; and for $A_3$, remove $v_5$ to get value $w(v_2)+w(v_7) \leq 0 \leq T$, and remove a vertex of minimum value among $v_2$ and $v_7$ to get value $w(v_5)+\max(w(v_2),w(v_7)) \geq T$.
This yields the desired EF[1,1] property.

Otherwise, it holds that
\[w(v_5) + \max(w(v_2),w(v_7)) < T.\]
In this case, consider the allocation $(A_1,A_2,A_3)$ defined by
\[A_1 = \{v_3,v_5,v_7\},~~A_2=\{v_2,v_4,v_6\},~~\mbox{and}~~A_3=\{v_1\}.\]
Again, the three sets are clearly independent in $C_7$. As before, it is enough to verify that, for each $i \in [3]$, the set $A_i$ can be brought to value at most $T$ and to value at least $T$ by removing at most one vertex.
For $A_1$, remove $v_3$ to get value $w(v_5)+w(v_7) \leq T$, and remove $v_7$ to get value $w(v_3)+w(v_5) \geq T$; for $A_2$, remove a vertex of maximum value among $v_4$ and $v_6$ to get value $w(v_2)+\min(w(v_4),w(v_6)) \leq \min(w(v_4),w(v_6)) \leq T$, and remove the vertex $v_2$ to get value $w(v_4)+w(v_6) \geq T$; and for $A_3$, remove $v_1$ to get value $0 \leq T$, and remove nothing to get value $w(v_1) \geq T$.
This completes the proof.
\end{proof}

\section{Algorithms for Fair Allocation under Conflict Constraints}\label{sec:algorithms}

In this section, we present a general framework for producing algorithms that find feasible fair allocations of graphs to agents with common preferences. To do so, we point out that the only non-constructive step in the proofs of our existential results lies in the strong coloring problem defined below.

\begin{definition}\label{def:strong_colo_prob}
Let $\calF$ be a collection of pairs, each consisting of a graph and a positive integer, and let $r$ be a positive integer.
In the {\em $(\calF,r)$-strong coloring problem}, the input consists of a pair $(G,\ell) \in \calF$ and a partition of the vertex set of $G$ into parts of size at most $\ell$, with at most $r$ parts of size strictly smaller than $\ell$, and the goal is to find a proper $\ell$-coloring of $G$ in which each color class intersects each part in at most one vertex.
\end{definition}
\noindent
Note that the existence of an algorithm for the $(\calF,r)$-strong coloring problem entails that every instance of the problem has a solution, equivalently, that every pair $(G,\ell) \in \calF$ satisfies $\ell \geq \chi_s^r(G)$.

We show that efficient algorithms for these strong coloring problems yield efficient fair allocation algorithms. The choice $r=1$ is appropriate for the fairness notions SD-EF1 and EF1, while the choice $r=2$ is appropriate for EF[1,1].

\begin{theorem}\label{thm:algo_gen}
Let $\calF$ be a collection of pairs, each consisting of a graph and a positive integer. 
\begin{enumerate}
  \item\label{itm:algo_1} If there exists a polynomial-time algorithm, deterministic or randomized, for the $(\calF,1)$-strong coloring problem, then there exists a polynomial-time algorithm, of the same type, for each of the following problems.
        \begin{enumerate}
            \item\label{itm:algSDGen} Given a pair $(G,\ell) \in \calF$ and a weak order $\succeq$ on the vertex set of $G$, find a feasible SD-EF1 allocation of $G$ to $\ell$ agents with respect to the common weak order $\succeq$.
            \item\label{itm:algEFGen} Given a pair $(G,\ell) \in \calF$ and a monotone additive valuation $w$ on the vertex set of $G$, find a feasible EF1 allocation of $G$ to $\ell$ agents with respect to the common valuation $w$.
        \end{enumerate}
  \item\label{itm:algo_2} If there exists a polynomial-time algorithm, deterministic or randomized, for the $(\calF,2)$-strong coloring problem, then there exists a polynomial-time algorithm, of the same type, for the problem of finding, given a pair $(G,\ell) \in \calF$ and an additive valuation $w$ on the vertex set of $G$, a feasible EF[1,1] allocation of $G$ to $\ell$ agents with respect to the common valuation $w$.
\end{enumerate}
\end{theorem}

\begin{proof}
We begin with the proof of Item~\ref{itm:algo_1}. Suppose that there exists a polynomial-time algorithm for the $(\calF,1)$-strong coloring problem. We prove Items~\ref{itm:algSDGen} and~\ref{itm:algEFGen} together, indicating the differences when needed. Consider an input that consists of a graph $G=(V,E)$ and an integer $\ell$ with $(G,\ell) \in \calF$ as well as either a weak order $\succeq$ on $V$ or a monotone additive valuation $w: \calP(V) \to \R$. In the latter case, we assume below that $w$ is non-decreasing, since if it is non-increasing, applying the algorithm with $-w$ gives the desired allocation.

We define an algorithm that, given such an input, permutes the vertices of $V$ non-increasingly with respect to either the weak order $\succeq$ or the valuation $w$, and produces an ordered partition $(Q_1, \ldots, Q_t)$ of $V$ into $t = \lceil |V|/\ell \rceil$ parts of consecutive vertices in the permutation, each of size exactly $\ell$ except possibly the part $Q_t$ that may have a smaller size. Then, the algorithm calls the assumed algorithm on the pair $(G,\ell) \in \calF$ and the partition $Q_1, \ldots, Q_t$ of $V$, which includes at most one part of size smaller than $\ell$, to obtain a proper $\ell$-coloring of $G$ in which each color class intersects each part in at most one vertex. The algorithm returns a sequence $\calA = (A_1, \ldots, A_\ell)$ of the color classes of this coloring, ordered arbitrarily.

The returned sequence $\calA$ is a feasible allocation of $G$ to $\ell$ agents, because the given algorithm returns a proper $\ell$-coloring of $G$. Moreover, since $|Q_j|= \ell$ for all $j \in [t-1]$, and since the $\ell$ color classes intersect each part in at most one vertex, we have $|A_i \cap Q_j| = 1$ for all $i \in [\ell]$ and $j \in [t-1]$, and $|A_i \cap Q_t| \leq 1$ for all $i \in [\ell]$. For Item~\ref{itm:algSDGen}, Lemma~\ref{lemma:SD_Pj} yields that $\calA$ is SD-EF1 with respect to the common weak order $\succeq$, and for Item~\ref{itm:algEFGen}, Lemma~\ref{lemma:EF1_Pj} yields that $\calA$ is EF1 with respect to the common non-decreasing valuation $w$. This ensures the correctness of the algorithm. As for the running time, since the given algorithm for the $(\calF,1)$-strong coloring problem runs in polynomial time, so does ours.

We turn to the similar proof of Item~\ref{itm:algo_2}. Suppose that there exists a polynomial-time algorithm for the $(\calF,2)$-strong coloring problem. Consider an input that consists of a graph $G=(V,E)$ and an integer $\ell$ with $(G,\ell) \in \calF$, along with an additive valuation $w: \calP(V) \to \R$. We may assume that $w$ is not monotone, as otherwise, the desired allocation can be obtained by the algorithm of Item~\ref{itm:algEFGen}, since every $(\calF,1)$-strong coloring instance is also an $(\calF,2)$-strong coloring instance.

We define an algorithm that, given such an input, partitions $V$ into two sets $V^+$ and $V^-$, consisting of goods and chores, respectively, under $w$, and permutes the vertices of $V^+$ non-increasingly with respect to $w$ and the vertices of $V^-$ non-increasingly with respect to $-w$. Then the algorithm produces ordered partitions $(Q^+_1, \ldots, Q^+_t)$ of $V^+$ and $(Q^-_1, \ldots, Q^-_{t'})$ of $V^-$ into parts of consecutive vertices in the permutations, each of size exactly $\ell$ except possibly the parts $Q^+_t$ and $Q^-_{t'}$ that may have smaller sizes. It calls the assumed algorithm on the pair $(G,\ell) \in \calF$ and the partition $Q^+_1, \ldots, Q^+_t,Q^-_1, \ldots, Q^-_{t'}$ of $V$, which includes at most two parts of size smaller than $\ell$, to obtain a proper $\ell$-coloring of $G$ in which each color class intersects each part in at most one vertex. Finally, the algorithm returns a sequence $\calA = (A_1, \ldots, A_\ell)$ of the color classes of this coloring, ordered arbitrarily.

The analysis of the algorithm is similar to that of Item~\ref{itm:algo_1}. The returned allocation $\calA$ is feasible, because its sets are the color classes of a proper $\ell$-coloring of $G$, and it is EF[1,1] with respect to the common valuation $w$ by Lemma~\ref{lemma:EF11_P_Q}. The polynomial running time follows from that of the given algorithm for the $(\calF,2)$-strong coloring problem. This completes the proof.
\end{proof}

We demonstrate the applicability of Theorem~\ref{thm:algo_gen} in two settings. The first concerns feasible fair allocations of general graphs with the number of agents depending on the maximum degree. Recall that our existential results, combined with Theorem~\ref{thm:3Delta-1}, imply that every graph with maximum degree $\Delta$ admits feasible SD-EF1, EF1, and EF[1,1] allocations with respect to every common preference of the relevant kind, whenever the number of agents is at least $3\Delta-1$. To obtain an algorithmic analogue of this result, with a somewhat larger threshold, we employ the following result of Harris~\cite{Harris23}.

\begin{theorem}[\cite{Harris23}]\label{thm:3Delta_algo}
For any fixed $\eps >0$, there exists a deterministic polynomial-time algorithm that, given a graph $G$ with maximum degree $\Delta$, a positive integer $\ell \geq (3+\eps) \Delta$, and a partition of the vertex set of $G$ into parts of size exactly $\ell$, finds a proper $\ell$-coloring of $G$ in which each color class intersects each part in exactly one vertex.
\end{theorem}

By combining Theorems~\ref{thm:algo_gen} and~\ref{thm:3Delta_algo}, we deduce that whenever the number of agents meets the above threshold, a feasible fair allocation with respect to common preferences can be found efficiently, as stated in Theorem~\ref{thm:intro_algos}.

\begin{proof}[ of Theorem~\ref{thm:intro_algos}]
Fix any $\eps >0$. Let $\calF$ be the family of all pairs $(G,\ell)$, where $G$ is a graph with maximum degree $\Delta$ and $\ell$ is a positive integer satisfying $\ell \geq (3+\eps)\Delta$. We claim that there exists a deterministic polynomial-time algorithm that, given a pair $(G,\ell) \in \calF$ and a partition of the vertex set of $G$ into parts of size at most $\ell$, finds a proper $\ell$-coloring of $G$ in which each color class intersects each part in at most one vertex. Indeed, given such an input, add to $G$ sufficiently many isolated vertices and use them to complete every part of size smaller than $\ell$ to a set of size exactly $\ell$. The resulting graph has the same maximum degree, and the resulting partition consists of parts of size exactly $\ell$. Then one can call the algorithm given in Theorem~\ref{thm:3Delta_algo} to obtain a proper $\ell$-coloring of the modified graph in which each color class intersects each part in exactly one vertex. Restricting this coloring to the original vertices of $G$ gives the desired coloring. In particular, the $(\calF,1)$- and $(\calF,2)$-strong coloring problems admit deterministic polynomial-time algorithms. The proof is completed by applying Theorem~\ref{thm:algo_gen}.
\end{proof}

Our second application of Theorem~\ref{thm:algo_gen} concerns the family of path graphs, addressed in~\cite{EGIKMNSVY26}, and generalizes a result of~\cite{KumarEGNV24}. We rely on the following algorithmic result for cycles due to Fleischner and Stiebitz~\cite{FleischnerS97}.

\begin{theorem}[\cite{FleischnerS97}]\label{thm:C_algo_4}
There exists a deterministic polynomial-time algorithm that, given a cycle graph $G$, an integer $\ell \geq 4$, and pairwise disjoint sets of vertices of $G$, each of size exactly $\ell$, finds a proper $\ell$-coloring of $G$ in which each color class intersects each set in exactly one vertex.
\end{theorem}

\begin{remark}
The above result is presented in~\cite{FleischnerS97} under the further assumption that no input set contains an edge of the cycle graph. This assumption can be avoided as follows. Subdivide every cycle edge whose endpoints lie in the same input set, apply the algorithm of~\cite{FleischnerS97} to the resulting cycle together with the same input sets, and restrict the obtained coloring to the original vertices. The restriction is proper, because the endpoints of every subdivided edge lie in the same input set and hence receive distinct colors.
\end{remark}

By combining Theorems~\ref{thm:algo_gen} and~\ref{thm:C_algo_4}, we obtain efficient algorithms for finding feasible fair allocations of path graphs to at least four agents with common preferences. Recall that the corresponding existential guarantees hold already for three agents.

\begin{theorem}\label{thm:paths_algo}
There exists a deterministic polynomial-time algorithm for each of the following problems.
\begin{enumerate}
  \item Given a path graph $G$, an integer $\ell \geq 4$, and a weak order $\succeq$ on the vertex set of $G$, find a feasible SD-EF1 allocation of $G$ to $\ell$ agents with respect to the common weak order $\succeq$.
  \item Given a path graph $G$, an integer $\ell \geq 4$, and a monotone additive valuation $w$ on the vertex set of $G$, find a feasible EF1 allocation of $G$ to $\ell$ agents with respect to the common valuation $w$.
  \item Given a path graph $G$, an integer $\ell \geq 4$, and an additive valuation $w$ on the vertex set of $G$, find a feasible EF[1,1] allocation of $G$ to $\ell$ agents with respect to the common valuation $w$.
\end{enumerate}
\end{theorem}

\begin{proof}
Let $\calF$ be the family of all pairs $(G,\ell)$, where $G$ is a path graph and $\ell$ is an integer satisfying $\ell \geq 4$. We claim that there exists a deterministic polynomial-time algorithm that, given a pair $(G,\ell) \in \calF$ and a partition of the vertex set of $G$ into parts of size at most $\ell$, finds a proper $\ell$-coloring of $G$ in which each color class intersects each part in at most one vertex. Indeed, given such an input, extend the path $G$ by attaching a path on sufficiently many new vertices to one of its endpoints, and use these new vertices to complete every part of size smaller than $\ell$ to a set of size exactly $\ell$. Let $G'$ be the cycle obtained by connecting the two endpoints of the resulting path, and note that the completed parts form a partition of the vertex set of $G'$ into parts of size exactly $\ell$. We may therefore apply the algorithm from Theorem~\ref{thm:C_algo_4} to obtain a proper $\ell$-coloring of $G'$ in which each color class intersects each completed part in exactly one vertex. Restricting this coloring to the original vertices of $G$ gives the desired coloring. In particular, the $(\calF,1)$- and $(\calF,2)$-strong coloring problems admit deterministic polynomial-time algorithms. The proof is completed by applying Theorem~\ref{thm:algo_gen}.
\end{proof}

\section{Concluding Remarks}\label{sec:remarks}

This paper studies fair allocation of indivisible items under conflict constraints for agents with common preferences, focusing on several standard fairness criteria. Our results reveal a close connection between these allocation problems and a hierarchy of strong chromatic numbers that we introduce. For stochastic-dominance envy-freeness up to one item (SD-EF1), where the agents are equipped with a common weak order, this connection is exact, as the strong chromatic number of type $1$ of the conflict graph precisely captures the threshold on the number of agents guaranteeing a feasible SD-EF1 allocation. For envy-freeness up to one item (EF1), when the agents share a monotone additive valuation, the same graph quantity yields a sufficient condition for a feasible EF1 allocation. For envy-freeness up to one item from each side (EF[1,1]), when the agents share an additive valuation that is not necessarily monotone, our sufficient condition for a feasible EF[1,1] allocation employs the strong chromatic number of type $2$. In contrast to the SD-EF1 situation, the sufficient conditions for EF1 and EF[1,1] are not necessary in general, suggesting a more delicate picture. A natural challenge is thus to obtain a refined graph-theoretic characterization in these two settings.

Another natural direction for future research would be to determine the numbers of agents that guarantee feasible fair allocations for a given conflict graph, when the agents are allowed to have heterogeneous preferences. For SD-EF1, we demonstrate that the strong chromatic number of type $1$ fails to provide an appropriate threshold already for two agents. For EF1 and EF[1,1], however, no analogous two-agent obstruction arises, because if a feasible bipartition is fair in both directions with respect to one agent's valuation, then, possibly after swapping the two parts, it is also fair for the other agent with respect to her own valuation. It would be intriguing to better understand how the diversity of the agents' preferences affects the number of agents needed to guarantee feasible fair allocations.

Finally, the link between fair allocation and strong colorability suggests considering the related notion of strong choosability in this context. For a positive integer $\ell$, a graph is called {\em strongly $\ell$-choosable} if for every partition of its vertex set into parts of size at most $\ell$, the graph obtained by adding a clique on each part is $\ell$-choosable, that is, every assignment of a list of $\ell$ colors to each vertex admits a proper coloring that assigns to each vertex a color from its list (see, e.g.,~\cite{Gutner92,AharoniBZ07}). Some results on strong colorability have choosability analogues. For instance, Fleischner and Stiebitz~\cite{FleischnerS92} established their `cycle plus triangles' theorem through its choosability counterpart, showing that the cycle $C_{3m}$ is even strongly $3$-choosable. Aharoni, Berger, and Ziv further proposed in~\cite{AharoniBZ07} a choosability variant of the `strong $2\Delta$-colorability' conjecture. Finding applications of strong choosability to fair allocation seems an interesting avenue to explore.

\bibliographystyle{abbrv}
\bibliography{fair_conflict}

\end{document}